\documentclass{amsart}

\usepackage{latexsym}
\usepackage{verbatim}
\usepackage{amsthm}

\usepackage{amssymb,amsmath,amscd,amsthm,lscape, graphicx, color, graphicx}

\usepackage[font=bf]{caption}
\usepackage{xcolor}                                           %
\usepackage[normalem]{ulem}
\renewcommand\sout{\bgroup\markoverwith
{\textcolor{red}{\rule[0.7ex]{3pt}{1.4pt}}}\ULon}

\newcommand{\ig}[1]{{}}

\newcommand{\CC}{\mathbb C}
\newcommand{\EE}{\mathbb E}

\newcommand{\RR}{\mathbb R}

\newcommand{\maE}{\mathcal E}

\newcommand{\maM}{\mathcal M}

\newcommand{\maR}{\mathcal R}

\newcommand\pa{\partial}
\newcommand\adj{\operatorname{ad}}
\newcommand\ede{\, := \,}
\newcommand\seq{\, = \,}

\newtheorem{theorem}{Theorem}[section]
\newtheorem{lemma}[theorem]{Lemma}
\newtheorem{proposition}[theorem]{Proposition}
\newtheorem{corollary}[theorem]{Corollary}

\theoremstyle{definition}

\newtheorem{remark}[theorem]{Remark}

\begin{document}

\title[Vol-of-vol expansion for SABR]{A volatility-of-volatility
  expansion of the option prices in the SABR stochastic volatility
  model}

\author[O. Grishchenko]{Olesya Grishchenko}
\email{olesya.grishchenko@gmail.com} \address{Division of Monetary
  Affairs, Federal Reserve Board, Washington, DC}

\author[X. Han]{Xiao Han} \email{xhan581@gmail.com} \address{IROM
  Department, McCombs School of Business, UT Austin, Austin, TX 78705
}

\author[V. Nistor]{Victor Nistor} \email{nistor@math.psu.edu}
\address{Pennsylvania State University, University Park, PA 16802, USA
  and Universit\'{e} Lorraine, 57000 Metz, France}

\thanks{The opinions expressed in this paper are those of the authors
  and do not necessarily reflect the views of the Board Governors of
  the Federal Reserve System or any other individuals within the
  Federal Reserve System. Manuscripts available from {\bf
    http:{\scriptsize//}www.math.psu.edu{\scriptsize/}nistor{\scriptsize/}.
} }

%V.N. was partially supported by the NSF Grant DMS-1016556.

\begin{abstract} 
We propose a general, very fast method to quickly approximate the
solution of a parabolic Partial Differential Equation (PDEs) with
explicit formulas. Our method also provides equaly fast approximations
of the derivatives of the solution, which is a challenge for many
other methods. Our approach is based on a computable series expansion
in terms of a ``small'' parameter. As an example, we treat in detail
the important case of the SABR PDE for $\beta = 1$, namely
$\partial_{\tau}u = \sigma^2 \big [ \frac{1}{2} (\partial^2_xu -
  \partial_xu) + \nu \rho \partial_x\partial_\sigma u + \frac{1}{2}
  \nu^2 \partial^2_\sigma u \, \big ] + \kappa (\theta - \sigma)
\pa_\sigma$, by choosing $\nu$ as small parameter. This yields $u =
u_0 + \nu u_1 + \nu^2 u_2 + \ldots$, with $u_j$ independent of
$\nu$. The terms $u_j$ are explicitly computable, which is also a
challenge for many other, related methods. Truncating this expansion
leads to computable approximations of $u$ that are in ``closed form,''
and hence can be evaluated very quickly. Most of the other related
methods use the ``time'' $\tau$ as a small parameter. The advantage of
our method is that it leads to shorter and hence easier to determine
and to generalize formulas. We obtain also an explicit expansion for
the implied volatility in the SABR model in terms of $\nu$, similar to
Hagan's formula, but including also the {\em mean reverting term.} We
provide several numerical tests that show the performance of our
method. In particular, we compare our formula to the one due to
Hagan. Our results also behave well when used for actual market data
and show the mean reverting property of the volatility.
\end{abstract}

\maketitle
\tableofcontents

\section*{Introduction}

We propose a new, general, non-iterative method to construct fast
approximations to suitable parabolic Partial Differential Equations by
constructing closed form approximate formulas for their solutions. In
addition to being very fast, our method has also the advantage that it
can be used to approximate also the derivatives of the solution. In a
nut-shell, our method is to first use a Dyson perturbative series to
expand the solution $u$ of our parabolic PDE in terms of a ``small''
parameter and then to explicitly {\em compute} the terms of the series
expansion. Dyson series expansions have long been used in
calculations; our achievement is to realize such a series with {\em
  computable} terms.

Our method expands the scope of the general perturbative-type method
devised and used in \cite{Wen1, CCMN, Labordere05}.  We illustrate our
method by providing closed form approximations to the solution $u$ of
the SABR partial differential equation (PDE) with mean reversion (the
$\lambda$SABR model) that arises in option pricing. Being able to
treat the mean reversion term in the SABR PDE is another advantage of
our method. Whereas in the aforementioned papers the small parameter
is the time, in this paper, the small parameter is the
``volatility-of-volatility'' parameter $\nu$ (see Equation
\eqref{eq.lambdaSABR}). This sets apart our approach from the previous
papers using series expansions, but see also \cite{Gatheral12a, Les03,
  Les15, Les17, Pascucci12, Pascucci17, Siyan}. The Dyson perturbative
series expansion is obtained by iterating Duhamel's formula, so we
will call it the {\em Duhamel-Dyson perturbative series expansion}.
See below for more on the earlier results, in general, and on the
Duhamel-Dyson perturbative series expansions, in particular.
  
We are interested in this paper in obtaining concrete results, so we
tried to keep the theoretical considerations (including proofs) to a
minimum, striving instead to carefully explain our method and to
obtain as quickly as possible our formulas, which we then tested
numerically.

\subsection*{The mathematical problem: the $\lambda$SABR PDE}
While our method is quite general, since we want to illustrate and
test our method through explicit results, we shall concentrate mainly
in this paper on the PDE
\begin{equation}\label{eq.lambdaSABR}
\begin{cases}
  \ \partial_{\tau}u \, = \, \sigma^2 \big [ \frac{1}{2}
    (\partial^2_xu - \partial_xu) + \nu \rho \partial_x\partial_\sigma
    u + \frac{1}{2} \nu^2 \partial^2_\sigma u \, \big ] \ + \ \kappa
  (\theta - \sigma) \pa_\sigma & \\
  \ u(x, \sigma, 0) = v(x, \sigma) \,. & \\
\end{cases}
\end{equation}
This PDE is obtained from the usual SABR PDE \cite{Les03, Les15} by
using the change of variables $x = \ln (S e^{rt})$ and $\tau = T - t$,
by fixing $\beta = 1$, and by including the standard mean reverting
term $\kappa (\theta - \sigma) \pa_\sigma$.  We shall use the term
$\lambda$SABR PDE for Equation~\eqref{eq.lambdaSABR}. When $\kappa =
0$, we shall also use the term SABR PDE.

In order to set up our perturbative method, we write $\kappa = \nu
\kappa_0$, for some fixed parameter $\kappa_0$.  That is, we consider
$\kappa$ of the same order as $\nu$. Therefore, when $\nu = 0$, the
SABR PDE reduces to the forward Black-Scholes PDE (for the forward
price and in $x = \ln (S e^{r\tau})$)
\begin{equation}\label{eq.BS.PDE}
 \pa_{\tau} u - \frac{1}{2}\sigma^2(\pa_x^2 - \pa_x) u= 0\,.
\end{equation} 
(here $\tau = T - t$). If $u_0$ is the solution of the Black-Scholes
equation, our method provides an expansion of the form
\begin{equation}\label{eq.taylor.u}
 u \, = \, u_0 + \nu u_1 + \nu^2 u_2 + \ldots + \nu^k u_k +
 O(\nu^{k+1}) \,.
\end{equation}

For general initial data $v$ in Equation \eqref{eq.lambdaSABR}, the
terms $u_j$ can be computed using an integration of $v$ with respect
to the Green functions, but this is not fully explicit (see Remarks
\ref{rem.kernel} and \ref{rem.bootstrap}). To obtain fully explicit
(or ``closed form'') results, we specialize our initial condition
$v(x, \sigma) = u(x, \sigma, 0)$ to the one of interest in practice,
that is, to $v =h$ given by
\begin{equation}\label{eq.def.h}
  h(x, \sigma) \seq |e\sp{x} - K|_{+} \ede \max\{ e\sp{x} - K, 0 \}\,.
\end{equation}
For this initial data, we also use $u = F_{SA}(S, K, \nu, \sigma,
\rho, \tau)$ to denote the solution of the $\lambda$SABR PDE, Equation
\eqref{eq.lambdaSABR}.  Let $F_{BS}$ be the Black-Scholes formula (or
function), which is the solution of the Black-Scholes PDE with initial
data $h = |e\sp{x} - K|_{+}$. Therefore,
\begin{equation}\label{eq.price.nu=0}
 F_{SA}(S, K, 0, \sigma, \rho, \tau) \, = \, F_{BS}(S, K, \sigma,
 \tau)\,,
\end{equation}
with $F_{BS}(S, K, \sigma, \tau)$ explicitly computable in terms of
the {\em cumulative normal distribution} function $N$.
%, which is a transcedental function. 
For the initial condition $v = h$, the expansion \eqref{eq.taylor.u}
becomes
\begin{equation}\label{eq.taylor}
 F_{SA}(S, K, \nu, \sigma, \rho, \tau) \, = \, F_{BS}(S, K, \sigma,
 \tau) + \nu F_1 + \nu^2 F_2 + \ldots + \nu^k F_k + O(\nu^{k+1}) \,.
\end{equation}
One of our main results is to provide a method to {\em explicitly
  compute} the coefficients $F_j$ in this expansion.  We thus obtain,
in principle, {\em closed form} approximations of the solution
$F_{SA}$.  We carry through the calculations to provide explicit,
closed-form formulas for $F_1$ and $F_2$.  This leads to an explicit,
closed-form approximation of $u = F_{SA}$ with error of the order
$O(\nu\sp{3})$.  See also the next subsection for a more precise and
complete description of our results.

Unlike the usual iterative numerical methods, our method based on
closed form solutions has a limited precision and hence a limited
range of applications (small $\nu$ and not too large $\tau$). However,
the type of explicit, closed form formulas that we obtained using our
method are favored in financial applications, where a great precision
is not needed, but it is important to have very fast, easy to
implement methods. In those applications, $F_{SA}(S, K, \nu, \sigma,
\rho, \tau)$ represents the no-arbitrage forward price of a European
Call Option with strike $K$, volatility $\sigma$, and time to expiry
$\tau$ in the SABR model. Nevertheless, a suitable modification of our
method can also be used to obtain approximations of the solution with
{\em an arbitrary predetermined precision} and {\em for any initial
  data.}  This uses an approximation of the Green function of the
$\lambda$SABR PDE that is of high order in time in combination with
the bootstrap procedure used in \cite{Wen1, CCMN}. The bootstrap
procedure is, however, iterative, it relies on numerical integration
and time discretization, which increase its cost, and is no longer
explicit. See Remarks \ref{rem.kernel} and \ref{rem.bootstrap}. The
numerical integration leads us to Fredholm integral operators, see
also \cite{Gatheral12}.

While, from a mathematical point of view, the SABR model {\em without
  mean reversion} is well understood since it fits into standard
theories, this is not the case once one includes the mean-reverting
term. For this reason, our results are less complete in the
mean-reverting case. For instance, when $\kappa = 0$, we do not need
boundary conditions at $\sigma = 0$, but when $\kappa > 0$, this is
not so clear. This paper was first circulated in 2014 as an SSRN
Preprint and is based on the second named author's 2012 Master Thesis
at Pennsylvania State University (under the supervision of
A. Mazzucato and V. Nistor) \cite{XiaoThesis}. In fact, the whole
project has started while all the authors were still at Pennsylvania
State University, whose continued support we are glad to acknowledge.

% As explained in the abstract of that paper (see also the abstract of
% the closely related paper \cite{Pascucci17}), one of the advantages
% of our approach (going back to \cite{CCMN}), is that it provides
% corrections to the Black-Scholes formula that are ``fully explicit
% and do not require any integration or any special functions.''

\subsection*{The method, history, and the main results}
Let us now explain our method, state our main results, and briefly
discuss how they compare with earlier works. As explained above, we
shall consider the $\lambda$SABR PDE, defined in Equation
\eqref{eq.lambdaSABR}, and our goal is to approximate its solution
$u$, in general, and the solution $u = F_{SA}(S, K, \nu, \sigma, \rho,
t$ for the initial data $u(0) = h = |e\sp{x} - K|_{+}$, in particular.
(From now on we shall use $t$ instead of $\tau$ in the $\lambda$SABR
PDE.)

To set up some notation and concepts, let us consider the general
evolution equation
\begin{equation}\label{eq.IVgen}
  \ \partial_t u = L u \,,\quad
  \ u(0) = v \,. 
\end{equation}
We shall formally write the solution of this equation as $u(t) =
e^{tL}v$.  Assume that $L = L_0 + V$. A typical result, expressing
$e^{tL}v$ in terms of $L_0$, $e^{sL_0}$, and $V$, is the following
{\em Duhamel-Dyson perturbative series expansion:}
\begin{equation}\label{eq.DuhamelDyson}
  e^{t L} v \seq I_0 + I_1 + \ldots + I_{n-1} + \maE_{n} \,,
\end{equation}
where\ $I_0 = e^{t L_0} v$,\ $I_1 = \int_{0}^{t} e^{(t-\tau)L_0}V
e^{(t-\tau)L_0}v \, d\tau$,\ and, in general,
\begin{equation}\label{eq.def.Ik}
I_k \ede \int_{0}^{t} \int_{0}^{\tau_{1}} \ldots \int_{0}^{\tau_{k-1}}
e^{(t-\tau_1)L_0}V e^{(\tau_1-\tau_2)L_0} \ldots V e^{(\tau_{k-1}
  -\tau_{k})L_0} V e^{\tau_{k}L_0} v \, d\tau \,,
\end{equation}
$d \tau = d\tau_1 \ldots d\tau_{k}$, $k < n$. A similar formula gives
the error term $\maE_n$, except that the last $L_0$ is replaced by
$L$, see Equation \eqref{eq.En}. This expansion was used also in
\cite{Les15} to study the distribution kernel of the SABR PDE.

The terms of the expansion in Equation \eqref{eq.DuhamelDyson} are
thus integrals of expressions of the form $e^{s_0 L_0} V e^{s_1 L_0} V
e^{s_2 L_0} \ldots e^{s_{n-1} L_0} V e^{s_n L_0}v$, where $s_0 + s_1 +
\ldots + s_n = t$, $s_j \ge 0$ (the exponentials and the $V$s
alternate). In general, we do not know any method to explicitly
compute integral of this form. However, if $L_0$ and $V$ generate a
finite dimensional Lie algebra, we can ``move'' all the exponentials
to one side of the formula such that the resulting integral is of a
form that can easily be computed. Concretely, this is achieved through
several applications of the Campbell-Hausdorff-Backer formula provided
by Theorem \ref{thm.finite.CHB}, as in \cite{Wen1, CCMN}. See
\cite{Siyan} for a form of this procedure that leads to solvable Lie
algebras (instead of nilpotent ones). In the specific case of the
$\lambda$SABR PDE, we obtain
\begin{equation}
  e^{(t-\tau_1)L_0}V e^{(\tau_1-\tau_2)L_0}V \ldots e^{(\tau_{k-1}
    -\tau_{k})L_0} V e^{\tau_{k}L_0} \seq e^{tL_0} P_k(t, \tau)\,,
\end{equation}
where $t \ge \tau_1 \ldots \ge \tau_k \ge 0$, $V$ and the exponentials
alternate, and $P_k(t, \tau)$ is {\em a polynomial} in $t$ and $\tau =
(\tau_1, \ldots, \tau_{k})$ with coefficients differential operators,
similar results, but taking $t$ as a small parameter were obtained in
\cite{Wen1, Wen2, CCMN, Pascucci17}. The integration over $\tau$ can
then be carried out in a straightforward manner since the
``exponential'' $e^{tL_0}$ (defined using the semi-group generated by
$L_0$) does not depend on $\tau$ anymore, and thus can be factored out
of the integral. The connection between the Duhamel-Dyson perturbative
series expansion and Lie algebra techniques was first noticed in
\cite{CCMN}.

Perturbative series expansions, in general, and the Duhamel-Dyson
perturbative series expansion, in particular, were used before in
\cite{Les15} for similar purposes in conjunction to the explicit
formulas for the heat kernel of the Laplacian on the hyperbolic plane
to study the SABR PDE. In general, perturbative expansions were very
extensively studied in physical applications. They were used in
\cite{Wen1, Wen2, CCMN, Pascucci17}, using the ``time'' $t$ as a small
parameter, to compute {\em in general} integrals of the form $I_k$ up
to an error of order $t^\infty$ (that is, to an arbitrary order in
$t$). An important progress in this direction was achieved in
\cite{Pascucci17}, where {\em complete} explicit calculations for the
SABR PDE were obtained (taking time as a small parameter). As that
paper shows, this is a very tedious, albeit elementary calculation. It
also undescores the need to find alternative perturbations using other
``small parameters,'' in particular, our calculations for the
$\lambda$SABR PDE are less tedious. See also \cite{Pascucci12,
  PascucciAnal, PascucciETF} for related calculations for the small
time asymptotics. Going further back in time, one needs to mention,
among many other contributions, the works of Henry-Labordere
\cite{Labordere05, Labordere07, LabordereBook} who used Riemannian
geometry heat kernel approximations, the works of Gatheral and his
collaborators who used heat kernel asymptotics to study the implied
volatility, and the works of Lesniewski and his collaborators
\cite{Les03, Les15, Les17}, who introduced and studied the SABR model,
the work of Fouque, Papanicolaou, Sircar, and Solna \cite{ FPSS03a,
  FPSS03b, FPSS11}, who studied stochastic volatility models and
singular perturbation techniques in option pricing, see also
\cite{FPSBook00}. For many of these authors, the motivation was to
study stochastic volatility models, which are discussed below in more
detail, and whose importance is underscored also by our results.

The explicit calculation of the terms $e^{tL_0} P_k(t, \tau)h$, albeit
elementary, becomes more and more tedious as $k$ growths. Because of
this, we only compute explicitly the second order approximation (in
$\nu$) of $F_{SA}$. That amounts to express $V$ as a multiple of $\nu$
and to collect the terms with the same powers of $\nu$ up to order
two, the remaining terms being included in the error. For the
coefficients of $\nu$ and $\nu^2$, we obtain complete, {\em closed
  form} (i.e. fully explicit) expression that we subsequently test
numerically. (The first term, more precisely, the coefficient of
$\nu^0$, is given by the Black-Scholes formula, so there is no
additional work to be done.)  Moreover, for the initial value $v = h =
|e^x -K|_+$, the calculation of $e^{tL_0} P_k(t, \tau)h$ greatly
simplifies, since $\pa_\sigma h = 0$.

Here is now an example of the type of results that we obtain. Let
\begin{equation}\label{eq.def.dpm}
	d_{\pm} \ede \frac{\ln (Se^{rt}/K)}{\sigma \sqrt{t}} \pm
        \frac{\sigma \sqrt{t}}{2} \seq \frac{ x - \ln K}{\sigma
          \sqrt{t}} \pm \frac{\sigma \sqrt{t}}{2}
\end{equation}
be the usual terms appearing in the Black-Scholes formula and $N$ be
the normal cumulative distribution function. Hence $N'(x) =
\frac{1}{\sqrt{2\pi}} e^{-x^2/2}$. Then, for the term $F_1$ appearing
in Equation \eqref{eq.taylor}, we obtain
\begin{equation}\label{eq.C1}
	F_1(S, K, \sigma, t) \seq \frac{K t}{2} \big (\kappa_0 (\theta
        - \sigma) \sqrt{t} - \rho \sigma d_{-} \big ) N'(d_{-})\, .
\end{equation}
The formula for $F_2$ is significantly more complicated, as seen in
Theorem \ref{theorem.main1}.

In practice, one is often interested in the ``implied volatility,''
or, more precisely, in the ``SABR model implied volatility.''  The
{\em SABR model implied volatility} $\sigma_{imp}$ is defined by
$F_{SA}(S, K, \nu, \sigma, \rho, t) \, = \, F_{BS}(S, K, \sigma_{imp},
t)$. (See also Equation \eqref{eq.imp.vol2} for the ``Black-Scholes
market implied volatility.'') The expansion \eqref{eq.taylor} then
yields an expansion in $\nu$ for the SABR model implied volatility,
namely,
\begin{equation}\label{eq.imp.vol}
   \sigma_{imp} \seq \sigma + \nu e_1 + \nu^2 e_2 + O(\nu^3) \,.
\end{equation}
The coefficients $e_1$ and $e_2$ are obtained by a tedious, but
elementary calculation. For example,
\begin{equation}\label{eq.formula.e1}
    e_1 \seq - \rho d_{-} \sigma \sqrt{t}/2
\end{equation}
% in general?
if $\kappa = 0$.  See Theorem \ref{theorem.main3} for the more
complicated formula for $e_2$. We also extend our approach to
approximate $\pa_S F_{SA}(S, K, \nu, \sigma, \rho, \tau)$, which is
important in practice, see Theorem \ref{theorem.main2}.

We thus obtain {\em two second order approximations} of $F_{SA}(S, K,
\nu, \sigma, \rho, t)$, namely,
\begin{equation}\label{eq.SABR.2}
    F_{SA,2}(S, K, \nu, \sigma, \rho, t) \, = \, F_{BS}(S, K, \sigma,
    t) + \nu F_1(S, K, \sigma, t) + \nu^2 F_2(S, K, \sigma, t)\,,
\end{equation}
obtained by truncating Equation \eqref{eq.taylor} and
\begin{equation}\label{eq.imp.2}
   F_{D}(S, K, \nu, \sigma, \rho, t) \, = \, F_{BS}(S, K, \sigma + \nu
   e_1 + \nu^2 e_2, t)\,,
\end{equation}
obtained by truncating the formula for $\sigma_{imp}$ to second order.
To numerically test our results, we compare these two approximations
with approximations obtained via other methods.  Thus, we first
compare our two approximation formulas with equation (2.17) in
\cite{Les03}, which yields a well known approximation of $F_{SA}$,
often used in practice. Our approximation agrees almost exactly with
Hagan's approximation when $t$ or $\nu$ are small. In fact, the power
series approximation in $y = \ln(F/K)$ of the first order coefficient
of $\nu$ in Hagan's formula coincides with our $e_1$ term. When $t$ is
taken one year, both approximation are very accurate, the largest
difference between the simulated price is 0.8 cents, or around
0.1--0.2\% of the option price, which is well within the bid-ask
spread. As $t$ increases to 10 years, the largest error increases to
1.5\%. However, note that the option also is more expensive for a
longer time to maturity, so the percentage (or relative) error is
still reasonable.
% For more details, see Figure \ref{fig:error} in
% Section \ref{sec.num}.
We have also compare our two approximate formulas with Hagan's formula
and with Monte-Carlo and Finite Difference approximations of the
option price across a range of strikes and maturities. Again, the
results are good, at least for $t \le 1$. See Sections
\ref{sec.market} and \ref{sec.num} for details. We note that in these
numerical tests, it is a challenge to perform an accurate enough test
using either Monte Carlo or Finite Differences, which shows the
importance of having alternative, faster methods. It is possible that
a modification of our method, combined with the results of
\cite{Wen1}, will allow us to extend the range of $\beta$.

\subsection*{Practical motivation}
The methods and results in this paper can be properly understood only
when put into the perspective of their applications. Let us say a few
words about this, without entering into unnecessary details.

Options and other derivatives have long been used in financial
applications. They were mathematically rigorously priced for the first
time by the famous Black-Scholes model \cite{BS}, under the assumption
that the underlying asset (for example, a stock) follows a log-normal
distribution, or geometric Brownian motion.  While the use of the
Black--Scholes model is widespread in practice, the Black-Scholes
model is known to produce option values significantly different from
the ones in the market.  The reasons for these differences are
studied, for instance, in Rubinstein's seminal paper
\cite{Rubinstein:85}. For us, the most important reason is that the
volatility of the underlying asset is non constant.

It is therefore increasingly common these days to consider models that
relax an assumption of the constant volatility of the underlying
asset. In practice, the time-varying volatility has lead to the
concept of {\em Black-Scholes market implied volatility} (sometimes,
simply, {\em implied volatility}) $\sigma_{imp}$, defined as the
volatility that would have to be used in the Black-Scholes formula to
recover (at a given time) the market price. That is,
\begin{equation}\label{eq.imp.vol2}
 C_{M}(S, K, t) \, = \, C_{BS}(S, K, \sigma_{imp}, t)\,,
\end{equation}
where $C_{M}(S, K, \tau)$ is given either by {\em market} prices or by
{\em another model} (such as the SABR model in our case). The forward
price $F$ is related to the actual price by the formula $F = e^{rt}
C$, where from now on, $t$ denotes the time to expiry. If the
geometric Brownian motion provided a perfect description of the
behavior of the underlying, then the Black-Scholes implied volatility
would be constant as a function of time to expiry $\tau$ and strike
$K$. However, the practice shows immediately that this is not the
case. See \cite{Gatheral:xx, GatheralStefanica17, Stefanica17} for
more information.

The fact that the volatility is not a constant in practice has
motivated the introduction of several other models. Among them, we are
interested in the {\em stochastic volatility models} (see
\cite{TyskSVol10, Les03, Heston:93, HW87}, for example). In stochastic
volatility models, the volatility is not only non-constant in time,
but is in fact a random variable with its own volatility, denoted here
$\nu$ and called {\em volatility of volatility} or {\em
  vol-of-vol}. This explains why the SABR PDE has two space variables
($x$ and $\sigma$ in our convention), unlike the Black-Scholes PDE,
which has only one space variable ($x$ in our convention).  One of the
disadvantages of stochastic volatility models, however, is their
greater computational cost.  It is important therefore to find good
closed form approximations of the solutions to option pricing in
stochastic volatility models. It is the purpose of this paper to
provide such an approximation, for small values of the {\em
  vol-of-vol} parameter $\nu$, by a Taylor-type expansion in $\nu$.
See \cite{FPSBook00, Pascucci17} for more on stochastic volatility
models and for further references to this topic.

We note that our results provide realistic estimates of the volatility
of volatility parameter $\nu$. They also show that that using
stochastic volatility improves the fit to the data, and, by further
including the mean reverting term, we get yet even better results. See
Subsection \ref{ssec.w.kappa} and, especially, Table \ref{table.3}.

\subsection*{Contents of the paper}
The paper is organized as follows. Section \ref{sec.DuhamelDyson}
contains the main results of the paper. In that section, we explain
our method in general and then we perform in detail the calculations
for the particular case of the $\lambda$SABR PDE, thus obtaining a
third order accurate in $\nu$ expansion of the solution. That amounts
to finding the exact formulas for $F_1$ and $F_2$ in the expansion
$F_{SA} = F_{BS} + \nu F_1 + \nu^2 F_2 + \ldots$. In this section, we
also provide an approximation for the distribution kernel of the
$\lambda$SABR PDE, which allows, in principle, to approximate its
solution to an arbitrary predetermined precision for any initial
data. One of the advantages of our method is that it allows us to
include the mean reverting term, so we keep it in our formulas for
most of our calculations. In Section \ref{sec.applications}, we use
extend our results by providing an asymptotic expansion for the
implied volatility $\sigma_{imp} = \sigma + \nu e_1 + \nu^2 e_2 +
\ldots$ obtained from the equation $F_{SA} = F_{BS}(\sigma_{imp})$.
We obtain a similar expansion for the $\Delta := \pa_S C$ hedging
parameter.  In Section \ref{sec.market} we use our two approximations
of the price, namely $F_{SA, 2} := F_{BS} + \nu F_1 + \nu^2 F_2 $ and
for $F_D := e^{rt}C_{BS}(\sigma + \nu e1 + \nu^2 e_2)$. We use these
formulas, as well as Hagan's formula to calibrate our model using
market data and to compare these models. Our approximations are seen
to be quite competitive in this regard. The results of the market
calibration also provide us with a realistic range of the parameters
for the numerical tests that we perform in Section \ref{sec.num}. In
that section, we perform several tests to see how well our results
approximate the exact solution of the $\lambda$SABR PDE. In
particular, we also compare our solutions to the one obtained using a
Finite Difference (FD) test. We discuss in detail the results of the
FD implementation and its challenges. In the last section, we outline
some possible extensions of our methods and results: a more precise
method to deal with the mean reverting factor and a method to treat
$\beta \neq 1$ in the original SABR PDE.

We thank Wen Cheng, Anna Mazzucato, Dan Pirjol, Camelia Pop, and Siyan
Zhang for useful discussions.

% By $A := B$ we shall mean that ```$A$ is defined by $B$.''

\section{The Duhamel-Dyson perturbative series expansion}
\label{sec.DuhamelDyson}

In this section, we introduce and explain in detail our method.  To
illustrate it, we perform some general calculations leading to the
second order approximation of the Green function of the $\lambda$SABR
PDE with initial condition $h$. This then leads to a determination of
the coefficients $F_1$ and $F_2$ in Equation \eqref{eq.taylor}, which
is our main approximation formula of the solution $F_{SA} = F_{BS} +
\nu F_1 + \nu^2 F_2 + \ldots$ of the $\lambda$SABR PDE.

As we are interested in this paper mainly in numerical methods for the
$\lambda$SABR PDE, we keep the theoretical aspects (including proofs),
to a minimum, trying to obtain as quickly as possible our formulas,
which we then test numerically.

\subsection{General results: The Duhamel and Campbell-Hausdorff-Backer formulas}
As in the Introduction, we cast our study of the SABR PDE as a
particular case of the problem of approximating general evolution
equations \eqref{eq.IVgen} (i.e. of the form $\pa_t u - Lu = 0$, $u(0)
= v$). In this subsection, we recall some general results pertaining
to the Equation \eqref{eq.IVgen}, including Duhamel's formula and its
generalization, the Duhamel-Dyson perturbative series expansion, as
well as the Campbell-Hausdorff-Backer formula. Beginning with the next
subsection, we specialize to the case of the $\lambda$SABR PDE.

The operator $L$ appearing in Equation \eqref{eq.IVgen} is called the
{\em parabolic generator} (of the given PDE or of the semi-group
$e^{tL}$).  Whenever the solution of the evolution problem
\eqref{eq.IVgen} exists, is unique, and depends continuously on the
initial data (in a suitable functions space $V$), we shall say that
our problem is {\em well posed} (in Hadamard's sense). In that case,
we shall write $u(t) = e^{tL}h$, where $e^{tL} : V \to V$ is a
continuous, linear map such that $e^{tL}v$ depends continuously on $v
\in V$ (i.e. $e^{tL}$ is a $c_0$-semi-group \cite{AmannBook, PazyBook,
  TucsnakBook}).  We shall often use this notation in a formal way,
that is, even if it was not fully justified mathematically.

Assume that $L = L_0 + V$ in our evolution equation, Equation
\eqref{eq.IVgen}, and that $u(0) = v$, as before. Assume that both
$L_0$ and $L$ generate $c_0$-semi-groups.  We then use $u(t) = e^{tL}
v$ and $\pa_t u - L_0 u = Vu$ to write {\em Duhamel's formula} in the
form
\begin{equation}\label{eq.Duhamel}
  e^{tL} v \seq e^{tL_0} v + \int_0^t e^{(t-\tau)L_0} V e^{\tau L} v
  \, d\tau\,.
\end{equation}
By substituting Duhamel's formula for $u(\tau)$, that is, for $t$
replaced by $\tau$, back in the last integral of the original
Duhamel's formula and then iterating this procedure $(n-1)$-times, we
obtain Equations \eqref{eq.DuhamelDyson} and \eqref{eq.def.Ik} of the
Introduction. The error term $\maE_n$ in Equation
\eqref{eq.DuhamelDyson} is the last integral of the resulting formula
and is similarly given by
\begin{equation}\label{eq.En}
\maE_n \ede \int_{0}^{t} \int_{0}^{\tau_{1}} \ldots
\int_{0}^{\tau_{n-1}} e^{(t-\tau_1)L_0}V e^{(\tau_1-\tau_2)L_0} \ldots
e^{(\tau_{n-1} -\tau_{n})L_0} V e^{\tau_{n}L} v \, d\tau \,.
\end{equation}
Notice that in the last exponential we have $L$ instead of $L_0$. If
we were to iterate one more time to obtain the Duhamel-Dyson formula
for $n+1$, we would substitute Duhamels's formula \eqref{eq.Duhamel}
applied to $e^{\tau_{n}L} v$ in the integral defining $\maE_n$.

In general, integrals of the kind defining $I_k$ are notoriously hard
to compute, and a lot of work has been devoted to understanding
them. However, if $L_0$ and $V$ generate a finite dimensional Lie
algebra, one can compute these integrals by establishing a
``Campbell-Hausdorff-Backer'' formula (CHB) \cite{Wen1, CCMN,
  Siyan}. The idea is to use the CHB formula to shift all the
exponential operators $e^{tL_0}$ to the left in our cases of interest.
This is what we are going to do in this paper as well. We note that in
our paper, the CHB formula reduces to a finite sum, so there are no
convergence issues.  It is possible to use the CHB formula also in
certain situations when the sum is not finite; see \cite{Siyan} for an
example.

Recall that $[T_1, T_2] := T_1T_2 - T_2T_1$ denotes the {\em
  commutator} of $T_1$ and $T_2$ and $ad_T(T_1) = [T, T_1]$. Clearly,
We are ready to state now one of our main technical ingredients, the
Campbell-Hausdorff-Backer (CHB) formula (see \cite{Wen1, CCMN, Siyan}
for a proof in a greater generality).

\begin{theorem}\label{thm.finite.CHB}
Let $L_0$ and $V$ be operators such that $\adj_{L_0}^{m+1}(V) = 0$ for
some $m$. For any $\theta > 0$, define
\begin{equation*}
  P_m(L_0, V; \theta) \ede \sum_{k=0}^m \frac{\theta^k}{k
    !}\adj^k_{L_0}(V) \seq L + \theta[L_0, V] +
  \frac{\theta^2}{2}[L_0, [L_0, V]] + \ldots .
\end{equation*}
Then
\begin{equation*}
  e^{\theta L_0} V \, = \, P_m(L_0, V; \theta) e^{\theta L_0},
\end{equation*}
and, similarly,
\begin{equation*}
  V e^{\theta L_0} \, = \, e^{\theta L_0} P_m(L_0, V; -\theta).
\end{equation*}
\end{theorem}

Let $\RR[\sigma, \pa_\sigma, \pa_x]$ be the algebra of differential
operators in $\pa_\sigma$ and $\pa_x$ with coefficients polynomials in
$\sigma$.  The CHB formula allows us to compute the terms $I_k$ of
Equation \eqref{eq.def.Ik} in the Duhamel-Dyson perturbative series
expansion, Equation \eqref{eq.DuhamelDyson}.  Indeed, a direct
calculation gives the following corollary.

\begin{corollary} \label{cor.Ik.polynomial}
With the assumptions of Theorem \ref{thm.finite.CHB}, we have 
\begin{multline*}
  e^{(t-\tau_1)L_0}V e^{(\tau_1-\tau_2)L_0} \ldots e^{(\tau_{k-1}
    -\tau_{k})L_0} V e^{\tau_{k}L}\\
   = e^{tL_0} P_m(L_0, V; -\tau_{1}) \ldots P_m(L_0, V;
   -\tau_{k-1})P_m(L_0, V; -\tau_k)
\end{multline*}
and hence $I_k = e^{tL_0} p_k(t)$, where $p_k(t) \in \RR[\sigma,
  \pa_\sigma, \pa_x]$ depends polynomialy on $t$, $L_0$, and $V$.
\end{corollary}

Similarly, by moving the exponentials to the right, 
we obtain that $I_k = q_k(t) e^{tL_0}$, where 
$q_k(t) \in \RR[\sigma, \pa_\sigma, \pa_x]$
depends polynomialy on $t$, $L_0$, and $V$. 

\subsection{Commutator calculations for SABR} 
We would like to use the CHB formula (Theorem \ref{thm.finite.CHB})
and its Corollary \ref{cor.Ik.polynomial} to identify explicitly the
terms $I_k$ in the Duhamel-Dyson perturbative series expansion for the
$\lambda$SABR PDE, namely in Equation \eqref{eq.DuhamelDyson}.  To
this end, we need to introduce the decomposition $L = L_0 + V$ and
compute the relevant commutators in order to see that we are in
position to use the CHB formula.

Let us introduce the differential operators:
\begin{equation}\label{eq.the.operators}
\begin{gathered}
  L_0 \, = \, \frac{1}{2} \sigma^2 (\partial^2_x - \partial_x)\,, \ \
  L_1 \, = \, \rho \sigma^2\partial_x\partial_\sigma \, + \,
  \kappa_0(\theta - \sigma)\partial_\sigma \,, \ \
  L_2 \, = \, \frac{1}{2}\sigma^2\partial^2_\sigma\,, \ \ \mbox{and}
  \\
 L \ede L_0 + \nu L_1 + \nu^2 L_2 \,.
\end{gathered}
\end{equation}
In this way, the $\lambda$SABR PDE (Equation \eqref{eq.lambdaSABR})
becomes a particular case of our general evolution equation, Equation
\eqref{eq.IVgen}, for $L \ede L_0 + \nu L_1 + \nu^2 L_2$. We can then
use the general results of the previous subsection, for $L = L_0 + V$,
where $V := \nu L_1 + \nu^2 L_2$.  The notation for the differential
operators $L_0$, $L_1$, $L_2$, $L$, and $V$ will remain fixed
throughout the paper.

The following lemma shows that, in the case of the $\lambda$SABR PDE,
we are in position to use the CHB formula an hence in position to {\em compute
  explicitly the terms in Duhamel-Dyson perturbative series expansion
  for $L_0$ and $V$}:

\begin{lemma} Let $b := \pa_x(\pa_x-1)$. Then \label{lemma.commutators}
$ \adj_{L_0}(L_1) = - \sigma [ \rho \sigma^2 \pa_x + \kappa_0 (\theta
    - \sigma)] b $ and $\adj^j_{L_0}(L_1) = 0$, if $j>1$. Similarly, $
  \adj_{L_0}(L_2) = -\frac{1}{2}(\sigma^2 + 2
  \sigma^3\partial_\sigma)b$, \ $ \adj^2_{L_0}(L_2) = \sigma^4 b^2, $
  and $\adj^j_{L_0}(L_2)=0$ if $j>2$.
\end{lemma}

In particular, $\adj^j_{L_0}(V) = 0$ if $j > 2$, and hence we can use
the CHB formula.

\subsection{Expansion in the powers of $\nu$}
Let us now expand the terms $I_k$ in the Duhamel-Dyson perturbative
series expansion in powers of $\nu$, see Equation
\eqref{eq.DuhamelDyson}. Then we use the CHB formula of Theorem
\ref{thm.finite.CHB} and its corollary, Corollary
\ref{cor.Ik.polynomial} (whose use was justified by Lemma
\ref{lemma.commutators}), to obtain

\begin{proposition}\label{prop.expansion}
We have
\begin{equation}
  e^{tL} \seq e^{tL_0} + \nu B_1 + \nu^2 B_2 + \ldots \nu^{n-1} B_{n-1} + 
  \nu^{n} \maR_n\,,
\end{equation}
where $B_j = P_j e^{tL_0} = e^{tL_0}Q_j$, $j < n$, with $P_j, Q_j \in
\RR[\sigma, \pa_\sigma, \pa_x]$it. Moreover, $\maR_n = P e^{tL_0} +
\maE_n = e^{tL_0} Q + \maE_n$, with $\maE_n$ as in Equation
\eqref{eq.En}.
\end{proposition}

We now turn to the calculation of the terms $B_j$, $j = 1, 2$.  The
goal is to ultimately obtain exact formulas for the terms $F_1= B_1 h$
and $F_2 = B_2 h$ in the expansion $F_{SA} = F_{BS} + \nu F_1 + \nu^2
F_2 + \ldots$ of Equation \eqref{eq.taylor}.

In view of Equations \eqref{eq.DuhamelDyson} and \eqref{eq.def.Ik} and
of the relation $V = \nu L_1 + \nu L_2$, let us introduce the notation
\begin{equation}\label{eq.def.Is}
\begin{gathered}
  J_1 \ede \int_0^te^{(t-\tau_1)L_0} L_1 e^{\tau_1L_0} \, d\tau_1 \,,
  \\
  J_2 \ede \int_0^te^{(t-\tau_1)L_0} L_2 e^{\tau_1L_0} \, d\tau_1 \,,
  \mbox{ and }\\
  I(T_1, T_2) \ede \int_0^t\int_0^{\tau_1} e^{(t-\tau_1)L_0} T_1
  e^{(\tau_1-\tau_2)L_0} T_2 e^{\tau_2L_0} \, d\tau_2\, d\tau_1
\end{gathered}
\end{equation}

The integrals above are defined as in \cite{CCMN, Siyan}. Also, let us
consider also the ``error term''
\begin{multline*}
 \maR_2 \, = \, I(L_1, L_2) + I(L_2, L_1) + \nu I(L_2, L_2) \\
 + \int_0^t \int_0^{\tau_1} \int_0^{\tau_2} e^{(t-\tau_1)L_0} V
 e^{(\tau_1-\tau_2)L_0} V e^{(\tau_2 - \tau_3)L_0} V e^{\tau_3L} \,
 d\tau_3 d\tau_2 d\tau_1.
\end{multline*}
(We notice again that the last exponential in the last formula is
different from the other ones, with $L$ being used instead of $L_0$.)
By collecting the terms with like powers of $\nu$ in the Duhamel-Dyson
perturbative series expansion for $L = L_0 + V$, where $L_0$, $V$ and
all the other differential operators are as introduced in Equation
\eqref{eq.the.operators}, we see that the expansion of Propositon
\ref{prop.expansion} becomes
\begin{equation}\label{eq.nu.series.J}
  e^{tL} \seq e^{tL_0} + \nu J_1 + \nu^2 (J_2 + I(L_1, L_1)) + \nu^3
  \maR_2 \, .
\end{equation}
That is,
\begin{equation}\label{eq.form.Bs}
   B_1 \seq J_1 \quad \mbox{ and } \quad B_2 \seq J_2 + I(L_1, L_1)
   \,.
\end{equation}
Note that Equation \eqref{eq.nu.series.J} is an exact formula -- not
just an approximation -- whenever the exponentials (or semi-groups)
are defined. Also, it is an equality of operators and hence
\begin{equation}\label{eq.exp.Bh}
  e^{tL}v
%  \seq e^{tL_0}v + \nu B_1v + \nu^2 B_2v + \nu^3 \maR_2 v 
  \seq e^{tL_0}v + \nu J_1v + \nu^2 (J_2v + I(L_1, L_1)) + \nu^3
  \maR_2 v \,,
\end{equation}
which we shall mainly use for $v(x) = h(x) := |e^x - K|_+$.

\subsection{Symbolic formulas for $B_j$, $j = 1, 2$.}
We now turn to a symbolic calculation of the terms $B_j$, $j = 1, 2$,
of Proposition \ref{prop.expansion}.  In view of Equation
\eqref{eq.form.Bs}, in order to compute $B_1$ and $B_2$, we need to
compute $J_1$, $J_2$, and $I(L_1, L_1)$.  To this end, we use the CHB
formula to successively move the exponentials to the left in the
integrals defining $J_1$, $J_2$, and $I(L_1, L_1)$. First, the
definitions of the integral $J_1$ in Equation \eqref{eq.def.Is} gives
\begin{equation}\label{eq.Duhamel.I1}
  J_1 \seq \int_0^t e^{tL_0}(L_1 -\tau \adj_{L_0}(L_1))\, d\tau\\ \seq
  e^{tL_0}\big ( tL_1 - \frac{t^2}{2} \adj_{L_0}(L_1) \big ) \,.
\end{equation}
Similarly, 
\begin{multline}\label{eq.Duhamel.I2}
  J_2 \seq \int_0^t e^{tL_0}(L_2 -\tau \adj_{L_0}(L_2) +
  \frac{\tau^2}{2} \adj_{L_0}^2(L_2))\\
  \seq e^{tL_0} \big (tL_2 - \frac{t^2}{2} \adj_{L_0}(L_2) +
  \frac{t^3}{6} \adj^2_{L_0}(L_2) \big )\, .
%
%  \seq \big (tL_2 + \frac{t^2}{2} \adj_{L_0}(L_2) + \frac{t^3}{6}
%  \adj^2_{L_0}(L_2) \big ) e^{tL_0} \,.
\end{multline}
Finally, the definition of the integral $I$ in Equation
\eqref{eq.def.Is} gives
\begin{equation}\label{eq.Duhamel.I}
\begin{gathered}
   I(L_1, L_1) \seq \int_0^t \int_0^{\tau_1} e^{(t-\tau_1) L_0} L_1
   e^{-\tau_1 L_0} (L_1 -\tau_2 \adj_{L_0}(L_1)) \, d\tau_2 d\tau_1 \\
   = \, \int_0^t \int_0^{\tau_1} e^{tL_0} (L_1 - \tau_1
   \adj_{L_0}(L_1))(L_1 -\tau_2 \adj_{L_0}(L_1))\, d\tau_2d\tau_1 \\
   = \, e^{tL_0} \Big ( \, \frac{t^2}{2}L_1^2 - \frac{t^3}{3}
   \adj_{L_0}(L_1) L_1 - \frac{t^3}{6} L_1\adj_{L_0}(L_1)+
   \frac{t^4}{8}(\adj_{L_0}(L_1))^2 \, \Big ) \,.
\end{gathered}
\end{equation}

This completes the symbolic calculations of $B_1$ and $B_2$, that is,
in terms of the exponentials $e^{tL_0}$ and the commutators of $L_j$.

\begin{remark}\label{rem.really}
Note that we obtained formulas for which all the exponentials of the
form $e^{tL_0}$ have been shifted to the right and formulas for which
they have all been shifted to the left.  The two forms are both
needed, but they serve different purposes.  The formulas that we have
obtained are explicit enough to allow for exact calculations.  In
particular, the coefficients $B_1 = J_1$ and $B_2 = J_2 + I(L_1, L_1)$
of Equation \eqref{eq.nu.series.J} can be written in the form $B_j =
P_j e^{tL_0}$ and $B_j = e^{tL_0}Q_j$, with $P_j, Q_j \in \RR[\sigma,
  \pa_\sigma, \pa_x]$, $j = 1, 2$, a fact that will be exploited
shortly.  The same is true for the other terms in the expansion ($j
\ge 3$).
\end{remark}

\begin{remark}\label{rem.wp} Let us notice that $e^{tL_0}$
is defined (since it corresponds to a family of heat equations).  In
particular,
\begin{equation*}
  e^{tL_0} \sigma^i \pa_x^j \seq \sigma^i \pa_x^j e^{tL_0} \,.
\end{equation*}
For $e^{tL}$, the situation is more complicated. Let $M := \RR \times
[0, \infty) \ni (x, \sigma)$. In case $\kappa_0 = 0$, that is, if
there is no mean reverting term, then $L$ is generated by the
derivatives $\sigma \pa_x$ and $\sigma \pa_\sigma$ together with
smooth, totally bounded functions (or coefficients) in $\sigma$ and
$x$. This type of differential operators were studied in
\cite{sobolev, aln1} in the framework of differential operators on
``Lie manifolds.'' In our case, $M$ is a Lie manifold for with the
structure generated by the standard compactification of $M$ to a disk
using hyperbolic coordinates. The geometry is that of the hyperbolic
plane, something noticed also in \cite{Les15}. Since Lie manifolds are
of bounded geometry \cite{sobolev}, the theory developed in
\cite{MazzucatoNistor1} and, more recently by Herbert Amann
\cite{AmannParab2, AmannMaxReg} shows that the equation
$\eqref{eq.IVgen}$ is well-posed in the Sobolev spaces associated to
this geometry, that is, in spaces of the form
\begin{equation*}
 H^m(M; g) \ede \{ f \in L^2(M) \, \vert \ \int_{M}
 \sigma^{i+j-1}\pa_x^i \pa_\sigma^j u\, dx d\sigma < \infty\,,\ i+j
 \le m\, \}\ .
\end{equation*}
This well-posedness result is, unfortunately, not known in the mean
reverting case. See however \cite{Les17, Siyan}. Because of this, our
treatment of the mean reverting case will be shorter and somewhat
formal. Analysis on this type of manifolds is related to that on
hyperbolic spaces \cite{Les15}. It is also related to the analysis of
edge singularities for PDEs on polyhedral domains \cite{BNZ3D1, CDN12,
  daugeBook, FeehanPop15, HengguangThesis}.
\end{remark}

\begin{remark}\label{rem.kernel} We can, in principle, compute exactly the 
distribution kernels (or Green functions) of the operators $B_j$. They
are very closely related to the kernels of the operators
$e^{tL_0}$. We will do that for $B_1$, see Remark \ref{rem.bootstrap}.
These kernels can then be used to integrate against {\em any initial
  data} $v$ to obtain, by truncating in the expansion in the small
parameter, a semi-discretization (i.e. a discretization only in time)
of our PDE. This semi-discretization can be turned into a full
discretization via an approximation of $v$ by functions in a finite
dimensional space and then by using numerical integration. This
approximation is very good for small time. It can be turned out into
an arbitrary precision approximation for {\em any time} using the
bootstrap method in \cite{Wen1, CCMN}.  See Remark \ref{rem.bootstrap}
for more details.
\end{remark}

%%%%%%%%%%%%%%%%%%%%%%%%%%%%%%%%%%%%%%%%%%%%%%%%%%%

\subsection{Heat kernels and convolutions}
\label{sec.int.ker}
The Green functions of the operators $B_j$ discussed in Remark
\ref{rem.kernel}, while explicit in principle, are given by pretty
complicated formulas. The calculation of $B_jh$, however, can further
be simplified using the specific properties of $h$ and $e^{tL_0}$.
Let us define
\begin{equation}\label{eq.def.phi_tau}
	\phi_{t}(x, \sigma) \seq \frac{1}{\sigma \sqrt{2\pi t}} \, e^{
          - \frac 12\big ( \frac{x}{\sigma \sqrt{t}} - \frac{\sigma
            \sqrt{t}}2 \big )^2 } \ .
\end{equation}

Then it is well known that
\begin{equation}\label{eq.fsol}
  e^{t L_0}(x,y) \seq \phi_t(x-y, \sigma) \,,
\end{equation}
in the sense that $e^{t L_{0}} f(x) = \int_{\RR}\, e^{t L_0}(x,y) f(y)
dy$ and $u(t, x, \sigma) := e^{t L_{0}} f(x)$ satisfies the partial
differential equation $\pa_t u(t) - L_0 u(t) = 0$ with initial
condition $u(0) = f$.  (That is, $e^{t L_0}(x,y)$ is the Green
function of $ \pa_t - L_0$.) In particular, this gives $e^{tL_0}h$,
which will be needed in the final formula, as follows.  Let $N(x) :=
(2 \pi)^{-1/2} \int_{-\infty}^x e^{-t^2/2} dt$ be the normal
cumulative function. Then
\begin{multline}\label{eq.eLh}
  e^{t L_0} h(x, \sigma)
  \seq \int_{\ln K}^\infty \phi_t(x-y, \sigma) (e^y - K) dy
  \seq
  e^x N ( d_+) - K N ( d_- )\, ,
%  \\ \seq \frac{1}{\sigma \sqrt{2 \pi t}} \, \int_{\ln K}
%e^{-\frac{(x - y -\frac{1}{2}\sigma^2t )^2}{2 \sigma^2 t}} (e^y - K)
%dy \\ \seq e^x N \Big ( \frac{x-\ln K}{\sigma \sqrt{t}} +
%\frac{\sigma \sqrt{t}}{2} \Big ) - K N \Big ( \frac{ x-\ln K}{\sigma
%\sqrt{t}} - \frac{\sigma \sqrt{t}}{2} \Big ) \seq e^x N ( d_+) - K N
%( d_- )\, ,
\end{multline}
where we used the notation of Equation \eqref{eq.def.dpm}.

This suggests to consider {\em integral kernel operators} $T_k$, where
$k(x, y)$ is a suitable measurable function and $T_kf(x) := \int_{\RR}
k(x, y)f(y)dy$. Then $k$ is called the {\em distribution kernel} of
$T_k$.  If $k(x, y) = \phi(x-y)$, then we shall write also $C_{\phi} =
T_{k}$ and we shall refer to $C_{\phi}$ as {\em the operator of
  convolution with $\phi$.} So
\begin{equation}
	C_{\phi} f(x) \seq \int_{\RR} \phi(x-y) f(y) dy.
\end{equation}
In particular, $ e^{ t L_0}(x,y)$ is the operator of convolution with
$\phi_t$:
\begin{equation}\label{eq.phi.exp}
	 e^{t L_0}(x,y) \seq C_{\phi_t} \, .
\end{equation}
All our convolution operators will be convolutions in the $x$ variable,
but they will often depend on $\sigma$ as a parameter. The parameter
$\sigma$ in the notation of the convolution will sometimes be omitted,
as in the last equation.

Let also $h(x) = |e^x - K|_{+} = (e^x - K)^+$ be the usual pay-off of
a European call option, Equation \eqref{eq.def.h}. 
Recall that a function $\phi : \RR \to \CC$ is said to be {\em of Schwartz 
class} if $x^k \pa_x^j \phi(x)$ is bounded for all $k, j \ge 0$. If $\phi$ is 
of Schwartz class, then $\phi$ and all its derivatives are integrable.
We have the following lemma, which underscores the important role 
played by the polynomial $b = \pa_x^2 - \pa_x$.

\begin{lemma}\label{lemma.kernel}
Assume $\phi$ is function of Schwartz class. Then
\begin{equation}
 \pa_x C_{\phi}  \seq C_\phi \pa_x = C_{\phi'} \,,
\end{equation}
in the sense that $(C_{\phi} \psi)' = C_\phi(\psi') =
C_{\phi'}(\psi)$, for $\psi$ of Schwartz class. If, moreover, there
exist two constants $C > 0$ and $\epsilon > 1$ such that
$|\phi^{(k)}(x)| \le C e^{-\epsilon |x|}$ for $k=0, 1, 2$, then
\begin{equation*} 
  b C_{\phi} h (x) \seq K \phi(x - \ln K) \,.
\end{equation*}  
In particular, if $P \in \RR[\sigma, \pa_x]$, the operators of the
form $P e^{tL_0}$ are convolution operators in the $x$ variable.
\end{lemma}

Note that $P e^{tL_0} = e^{tL_0}P$ if $P \in \RR[\sigma, \pa_x]$.
More precisely, we obtain that $P e^{tL_0} = e^{tL_0}P$ is a {\em
  family parametrized by $\sigma$} of convolution operators in the $x$
variable.

\begin{proof}
The proof is by direct calculation. For the convenience of the reader,
we note that the second relation follows from the first relation and
from $bh = (\pa_x^2 - \pa_x) h = K \delta_{\ln K}$, where $\delta_{y}$
is the Dirac distribution concentrated at $y$.
\end{proof}

The following remark explains the idea of the last step of the
calculation yielding the terms $F_1= B_1 h$ and $F_2 = B_2 h$ in the
expansion $F_{SA} = F_{BS} + \nu F_1 + \nu^2 F_2 + \ldots$ of Equation
\eqref{eq.taylor}.

\begin{remark}\label{rem.convolution}
Let us first notice that
\begin{equation}\label{eq.deriv.phi_tau}
	\pa_x^n \phi_{t}(x) \, = \, \tilde H_n(x) \, \phi_{t}(x)\,,
\end{equation}
for some polynomials $\tilde H_n(x)$, $g \ge 0$, whose coefficients
are functions of $\sigma$. Let $H_n$ be the Hermite polynomials and
\begin{equation}\label{eq.def.ell}
	\ell(x) \ede \frac{x}{\sigma \sqrt{t}} - \frac{\sigma
          \sqrt{t}}2 \,.
\end{equation}
Then $\tilde H_n = (- \frac{1}{\sigma \sqrt{t}})^n H_n(\ell(x))$. See
Remark \ref{rem.Hermite} for more details.  Let $B_j = e^{tL_0} Q_j$,
as in Proposition \ref{prop.expansion}, and let $b f= \pa_x^2 -
\pa_x$, as before.  Then $Q_j = \sum_i G_{ji} \pa_\sigma^i$ (a finite
sum), with $G_{ji} \in \RR[\sigma, \pa_x]$. We will see that $G_{j0} =
\tilde G_j b$, with $\tilde G_j = \sum_i a_{ji} \pa_x^i$, where
$a_{ji} = a_{ji}(\sigma)$ are polynomials in $\sigma$. Recall also
that $\pa_x$, $\sigma$, and $e^{tL_0}$ commute, and hence $e^{tL_0}$
and $G_{j0}$ commute.  Lemma \ref{lemma.kernel} then simplifies the
calculation of $B_j h$, $j = 1, 2$, as follows:
\begin{multline}\label{eq.convolution}
 B_j h (x)\seq e^{tL_0} \big (\sum_i G_{ji} \pa_\sigma^i \big ) h (x)
 \seq e^{tL_0} G_{j0} h (x) \seq G_{j0} e^{tL_0} h (x) \\
 \seq b \tilde G_j e^{tL_0} h (x) \seq b C_{\tilde G_{j} \phi_t} h
 \seq K \tilde G_{j} \phi_t (x - \ln K, \sigma)\\
 \seq \sum_i a_{ji}(\sigma) \tilde H_i(x - \ln K) \phi_t(x - \ln K,
 \sigma) \,.
\end{multline}
In order to complete our calculation, all that is thus left to do is
to find the differential operators $\tilde G_j$, $j = 1, 2$, or,
equivalently, the polynomials $a_{ji}(\sigma)$. We also notice that
$\pa_x (\psi(x - a)) = (\pa \psi)(x -a)$, so there is no danger of
confusion when writing $\tilde G_{j} \phi_t (x - \ln K, \sigma)$.
\end{remark}

The last remark explains also why, if we are interested in evaluating
terms of the form $B_j h$, it is more convenient to keep the
exponential to the left of the formula (to get rid of the
$\pa_\sigma$).

This discussion allows us to identify the distribution kernel of
$B_1$.

\begin{remark}\label{rem.bootstrap} Recall that
$B_1 = J_1 = e^{tL_0}\big ( tL_1 - \frac{t^2}{2} \adj_{L_0}(L_1)
  \big)$, which combines the formula for $J_1$, Equation
  \eqref{eq.def.Is} with the identification \eqref{eq.form.Bs}. Then
\begin{multline*}
 B_1 \seq e^{tL_0}\Big ( t \big ( \rho
 \sigma^2\partial_x\partial_\sigma \, + \, \kappa_0(\theta -
 \sigma)\partial_\sigma \big ) \ + \ \frac{t^2\sigma }{2} [ \rho
   \sigma^2 \pa_x + \kappa_0 (\theta - \sigma)] (\pa_x^2 - \pa_x) \Big
 )\\
 \seq t \big ( \rho \sigma^2\partial_x + \kappa_0(\theta - \sigma)
 \big ) e^{tL_0} \pa_\sigma \, + \, \frac{t^2\sigma }{2} [ \rho
   \sigma^2 \pa_x + \kappa_0 (\theta - \sigma)] (\pa_x^2 - \pa_x)
 e^{tL_0}\\
 \seq t \kappa_0(\theta - \sigma) C_{\phi_t} \pa_\sigma \, + \, t\rho
 \sigma^2 C_{\tilde H_1 \phi_t} \pa_\sigma \, - \,
 \frac{t^2 \kappa_0 \sigma (\theta - \sigma)}{2}  C_{\tilde H_1 \phi_t}\\
 \, + \, \frac{t^2\sigma}{2} [ \kappa_0 (\theta - \sigma) - \rho
   \sigma^2 ] C_{\tilde H_2 \phi_t} \, + \, \frac{t^2\sigma^3\rho}{2}
 C_{\tilde H_3 \phi_t} \,.\\
\end{multline*}
This gives $B_1 v = C_{\psi_1} \pa_\sigma v + C_{\psi_2}v$.  Then the
solution of any initial value problem with $u(0) = v$ can be
approximated by $u\approx D_t v$,
\begin{equation*}
  D_tv(x, \sigma) \ede \int \Big [ \big (\phi_t(x - y, \sigma) +
    \psi_1(x -y, \sigma) \big ) v(y, \sigma)\, + \, \psi_2(x- y,
    \sigma) \pa_\sigma v(y, \sigma)\, \Big ] \, dy \, .
\end{equation*}
Of course, a better estimate will be obtained if one includes also the
$B_2$ term, which is, however, much more complicated. One can perform
a bootstrap, that is, divide the interval $[0, T]$ into $N$
subintervals, and solve on each interval the corresponding initial
problem approximately. That is
\begin{equation*}
  u((k+1)T/N) \approx D_{T/N} u(kT/N)\,.
\end{equation*}
This method works well once one establishes an error $\| u(t) - D_t v
\| = O(t^a)$ with $a > 1$.  See \cite{Wen1, CCMN}.
\end{remark}

\subsection{Calculation of the differential operators $\tilde G_j$}
We now pursue in detail the method outlined in Remark
\ref{rem.convolution} to compute explicitly closed form expressions
for $F_1 = B_1 h$ and $F_2 = B_2h$, where, we recall, $B_1 = J_1$ and
$B_2 = J_2 + I(L_1, L_1)$.  The differential operators $\tilde G_j \in
\RR[\sigma, \pa_x]$, $j = 1, 2$, considered in this subsection are as
defined in that remark and thus satisfy $B_j h (x) = K \tilde G_{j}
\phi_t (x - \ln K, \sigma)$.

We continue to use the notation $b = \pa_x^2 - \pa_x$, for
simplicity. Recall also from Lemma \ref{lemma.commutators} that
$\adj_{L_0}(L_1) = - \sigma [\rho \sigma^2 \pa_x + \kappa_0 (\theta -
  \sigma)] b$. Lemma \ref{lemma.kernel} and Equation
\ref{eq.Duhamel.I1} then give (since $t=0$)
\begin{equation*}
  B_1h \seq J_1h \seq e^{tL_0}\big ( tL_1 - \frac{t^2}{2}
  \adj_{L_0}(L_1) \big ) h \seq - \frac{t^2}{2} e^{tL_0}
  \adj_{L_0}(L_1) h\,.
\end{equation*}
Hence $\tilde G_1 =: \frac{t^2 \sigma}{2} [\rho \sigma^2 \pa_x +
  \kappa_0 (\theta - \sigma) ] = a_{10} + a_{11} \pa_x $ and
\begin{multline}\label{eq.final.B1}
  F_1(t, x, \sigma) \seq B_1h (t, x, \sigma) \seq K \tilde G_1
  \phi_t(x - \ln K, \sigma)\\
    \seq K \big ( a_{10} + a_{11} \tilde H_1(x - \ln K) \big )
    \phi_t(x - \ln K, \sigma) \,,
\end{multline}
which yields right away the formula \eqref{eq.C1} in the Introduction.

The second term is similar, but the calculation requires more work.
We shall use that $L_1 h = L_2 h = 0$. We first calculate the
intermediate term
\begin{multline}
  L_1 [L_0, L_1]h \seq - (\rho \sigma^2\partial_x\partial_\sigma +
  \kappa_0(\theta - \sigma)\partial_\sigma)\sigma [\rho \sigma^2 \pa_x
    + \kappa_0 (\theta - \sigma)] bh\\
      \seq - (\rho \sigma^2\partial_x + \kappa_0(\theta - \sigma)) [ 3
        \rho \sigma^2 \pa_x + \kappa_0 (\theta - 2\sigma)] bh \\
  \seq - \big [ \rho^2 \sigma^4\partial_x^2 + \kappa_0 \rho \sigma^2
    (4\theta - 5\sigma)\pa_x + \kappa_0 (\theta - \sigma)(\theta -
    2\sigma) \big ] bh \,.
\end{multline}
We then use this calculation together with $[L_0, L_1]^2 h \seq
\sigma^2 [\rho \sigma^2 \pa_x + \kappa_0 (\theta - \sigma)]^2 b^2 h$,
the formulas \eqref{eq.Duhamel.I2} and \eqref{eq.Duhamel.I}, the
relation $\pa_\sigma h = 0$, and the fact that $\pa_\sigma$ commutes
with $\pa_x$, to obtain
\begin{multline*}
  B_2h \seq (J_2 + I(L_1, L_1)) h \seq e^{tL_0} \big (tL_2 -
  \frac{t^2}{2} \adj_{L_0}(L_2) + \frac{t^3}{6} \adj^2_{L_0}(L_2) \big
  )h \\
  + e^{tL_0} \Big ( \, \frac{t^2}{2}L_1^2 - \frac{t^3}{3}
  \adj_{L_0}(L_1) L_1 - \frac{t^3}{6} L_1\adj_{L_0}(L_1)+
  \frac{t^4}{8}(\adj_{L_0}(L_1))^2 \, \Big )h \\
   \seq e^{tL_0} \Big (- \frac{t^2}{2} \adj_{L_0}(L_2) + \frac{t^3}{6}
   \adj^2_{L_0}(L_2)
  - \frac{t^3}{6} L_1\adj_{L_0}(L_1)+ \frac{t^4}{8}(\adj_{L_0}(L_1))^2
  \, \Big )h \\
  \seq e^{tL_0} \Big (\, \frac{t^2\sigma^2}{4} + \frac{t^3 \sigma^4}{6}b
    + \frac{t^3}{6} \big [ 3\rho^2 \sigma^4\partial_x^2 + 
  \kappa_0 \rho \sigma^2 (4\theta - 5\sigma)\pa_x + \kappa_0^2 (\theta - \sigma)(\theta - 2\sigma) \big ]\\
   + \frac{t^4\sigma^2}{8} [\rho^2 \sigma^4 \pa_x^2 + 2 \rho \sigma^2
     \kappa_0 (\theta - \sigma)\pa_x + \kappa_0^2 (\theta - \sigma)^2]
   b \, \Big )b h \,.
\end{multline*}

Let
\begin{align*}
    a_{20} & \seq \frac{t^2\sigma^2}{4} +
    \frac{t^3\kappa_0^2}{6}(\theta - \sigma)(\theta - 2\sigma)\\
    a_{21} & \seq -\frac{t^3\sigma^4}{6} + \frac{t^3 \kappa_0\rho
      \sigma^2 }{6} (4\theta - 5\sigma)
    -\frac{t^4\kappa_0^2\sigma^2}{8}(\theta - \sigma)^2\\
   a_{22} & \seq \frac{t^3\sigma^4}{6} + \frac{t^3\rho^2\sigma^4}{2} +
   \frac{t^4\kappa_0^2\sigma^2}{8}(\theta - \sigma)^2 -
   \frac{t^4\kappa_0 \rho \sigma^4}{4}(\theta - \sigma)\\
   a_{23} & \seq \frac{t^4\kappa_0 \rho \sigma^4}{4}(\theta - \sigma)
   - \frac{t^4 \rho^2 \sigma^6}{8}\\
   a_{24} & \seq \frac{t^4 \rho^2 \sigma^6}{8}\,.
\end{align*}
Then $\tilde G_2 = \sum a_{2i} \pa_x^i$ and, finally,
\begin{multline}\label{eq.final.B2}
  F_2(t, x, \sigma) \seq B_2h (t, x, \sigma) \seq K \tilde G_2
  \phi_t(x - \ln K, \sigma)\\
    \seq K\, \sum_{j=0}^4\, a_{2i} \tilde H_i(x - \ln K) \phi_t(x -
    \ln K, \sigma)\,.
\end{multline}

Putting together our calculations, we obtain the following result.
(Recall the notation of Equation \eqref{eq.taylor}, but see also
Equation \eqref{eq.SABR.2}.)

\begin{theorem}\label{theorem.main1}
The formal second order approximation of solution $F$ of the
$\lambda$SABR PDE is
\begin{equation}\label{eq.option.price1}
  	F_{SA,2}(S, K, \nu, \sigma, \rho, t) \, = \, e^{tL_0}h(t, x,
        \sigma) + \nu B_1 h(t, x, \sigma) + \nu^2 B_2 h(t, x, \sigma)
        \,,
\end{equation}
where $e^{t L_0} h$ is given by Equation \eqref{eq.eLh}, $B_1h = F_1$
and $B_2h = F_2$ are given explicitly by Equations \eqref{eq.final.B1}
and \eqref{eq.final.B2}, and $x = \ln(Se^{rt})$, as before.
\end{theorem}

\begin{remark}
Since the initial data $h(x) = |e^x - K|_{+} = (e^x - K)^{+}$ is very
smooth in $\sigma$ -- in fact, even independent of $\sigma$ in our
stochastic volatility models -- the methods of \cite{CCMN} will give
that $\maR h$ is bounded (and hence that $\nu^3 \maR h$ is very small
for $\nu$ small). That is,
\begin{equation}\label{eq.nu3}
  	F_{SA}(S, K, \nu, \sigma, \rho, t) - F_{SA,2}(S, K, \nu,
        \sigma, \rho, t) \, = \, O(\nu^3) \,,
\end{equation}
A rigorous proof of this fact is beyond the scope of this paper,
however.
\end{remark}

Similar formulas for $\kappa = 0$ but for general $\beta$ were
obtained in \cite{Pascucci17}. They use a small time asymptotic
method, very similar to the one in \cite{CCMN}, which leads to much
longer formulas.

\subsection{Remarks on the implementation}

When implementing the formula for $F_{SA,2}(S, K, \nu, \sigma, \rho,
t)$, it will be convenient to take into account the following
remarks. For simplicity, we restrict to the case $\kappa_0 = 0$ from
now on.

\begin{remark} \label{rem.BS}
First of all, for our calculations, it was convenient to work with the
{\em forward prices}, usually denoted $F$ and decorated with various
indices.  In practice, however, one may need to use the actual prices,
denoted $C$ and decorated with the corresponding indices. They are
related by the formula $F = e^{rt}C$, which corresponds to the fact
that the price $C$ is given in ``today's'' currency, whereas the
forward price is quoted in using the value of the currency at the
expiration. Here, of course, $r$ is the interest rate and $t$ is the
time to expiration, as before. This is in agreement with our
calculations. Indeed, let $d_{\pm} := \frac{\ln S-\ln K + rt}{\sigma
  \sqrt{t}} \pm \frac{\sigma \sqrt{t}}{2} := \frac{\ln (F/K) }{\sigma
  \sqrt{t}} \pm \frac{\sigma \sqrt{t}}{2}$, as before, Equation
\eqref{eq.def.dpm}.  In particular, $\ell(x - \ln K) = d_{-}$. We
recall then that the Black-Scholes formula for the price $C_{BS}(S, K,
\sigma, t)$ of a call option with strike $K$, underlying $S$, and
volatility $\sigma$ is given by
\begin{equation}\label{eq.BS}
	C_{BS}(S, K, \sigma, t) = S N (d_{+} ) - e^{-rt} K N ( d_{-} )
        \,.
\end{equation}
Recall that $F_{BS} = e^{tL_0}h$, and hence the formula $F_{BS} =
e^{rt}C_{BS}$ is verified, by Equation \eqref{eq.eLh}, since $x - \ln
K= \ln (S e^{rt}/K)$.
\end{remark}

The term $\ln (S e^{rt}/K)$ is called the {\em log-moneyness}.

\begin{remark}\label{rem.Hermite} Recall that the 
Hermite polynomials $H_n(x)$ (the probabilist's version), are given by
\begin{equation}\label{eq.def.Hermite}
	H_n(x) \, : = \, (-1)^n e^{x^2/2} \pa_x^n e^{-x^2/2}\, \quad n
        \ge 0 \,.
\end{equation}
Of these, only $H_0(x) = 1$, $H_1(x) = x$, $H_2(x) = x^2 -1$, $H_3(x)
= x^3 - 3x$, and $H_4(x) = x^4 - 6x^2 + 3$ are needed for the
evaluation of $B_1$ and $B_2$.  If $\ell(x)$ is any linear function of
$x$ (so $\ell'$ is a constant), then
\begin{equation*}
	\pa_x^n e^{- \ell(x)^2/2} \, = \, (-\ell')^n H_n(\ell(x))\,
        e^{-\ell(x)^2/2} \,.
\end{equation*}
Let us fix $\ell(x) := \frac{x}{\sigma \sqrt{t}} - \frac{\sigma
  \sqrt{t}}2$ from now on, as in \eqref{eq.def.ell}, which gives
$\phi_t(x, \sigma) = c e^{-\ell(x)^2/2}$, $c \in \RR$, and hence we
have $\tilde H_n(x) := (-\ell')^n H_n(\ell(x)) $. Recall that $\phi_t$
is given by Equation \eqref{eq.def.phi_tau} (but see also Equation
\eqref{eq.phi.exp}) and the polynomials $\tilde H_n$ are as defined in
Equation \eqref{eq.deriv.phi_tau} (that is, $\pa_x^n \phi_{t}(x) =
\tilde H_n(x) \, \phi_{t}(x)$).  Furthermore, in implementation, it
will be useful to notice that $\ell(x - \ln K) = d_-$, that $\tilde
H_n(x - \ln K) = \big ( - \frac{1}{\sigma \sqrt{t}} \big )^nH_n(d_-)$,
and that $\phi_t(x - ln K, \sigma) = \frac{1}{\sigma \sqrt{2\pi t}}
e^{-d_-^2/2}.$ Recall that $d_-$ was defined in Equation
\eqref{eq.def.dpm}.
\end{remark}

\begin{remark}\label{rem.homogeneous}
Let $f(S, K)$ be an arbitrary function. We shall say that it is {\em
  homogeneous of degree one in $(S, K)$} if $f( \lambda S, \lambda K)
= \lambda f(S, K)$. If that is the case, the function $g = K^{-1} f$
can be written solely in terms of $y = \ln (Se^{rt}/K)$ (we assume $r$
and $t$ to be parameters in this discussion).  For instance, the
Black-Scholes pricing formula is homogeneous of degree one in $(S,
K)$. The form of the Equation \eqref{eq.lambdaSABR} tells us that the
forward price $F_{SA}$ of a European call in the SABR model is also
homogeneous of degree one in $(S, K)$. Since $C_{SA} := e^{-rt}
F_{SA}$, we have that $C_{SA}$ is also homogeneous of degree one in
$(S, K)$. We also see that $B_1h$ and $B_2h$ are homogeneous of degree
one in $(S, K)$ as well, which is consistent with the properties of
$F_{SA}$. In particular, $F_{SA,2}$ and $C_{SA} := e^{-rt} F_{SA}$ are
also homogeneous in $(S, K)$ of degree one.
\end{remark}

\begin{remark}\label{rem.RelBS}
Let $y := x - \ln K$ denote the log-moneyness. We have already noticed
that $\ell(y) = d_{-}$. Let
\begin{equation}\label{eq.def.C_rel}
  C_{rel}(y, \sigma, t) \, := \, K^{-1}e^{rt}C_{BS}(S, K, \sigma, t)
  \,,
\end{equation}
so that $C_{rel}$ depends directly only on $y$ (and only indirectly on
$S$ or $K$, as explained in Remark \ref{rem.homogeneous}). This will
be useful in implementation, where often we will be given $y$, but not
$S$ and $K$, and we will want to estimate relative errors of the form
 \begin{equation}
  \ln\, \frac{C_{BS}(S_1, K, \sigma, t)}{C_{BS}(S_2, K, \sigma, t)} =
  \ln\, \frac{C_{rel}(y_1,\sigma, t)}{C_{rel}(y_2 \sigma, t)}\,,
 \end{equation}
 with $y_j = ln(S_je^{rt}/K)$. Also for the purpose of implementation,
 let us notice that we have
 \begin{equation}
  C_{SA,2}(S, K, \nu, \sigma, \rho, t) \, = \, Ke^{-rt} \Big (
  C_{rel}(y, \sigma, t) + ( \nu \frac{\rho \sqrt{t}}{4\sqrt{2\pi}} +
  \nu^2 \Xi ) e^{- d_-^2/2} \Big )\,,
 \end{equation}
 where
 \begin{equation}
  \Xi = \frac{\sigma t^{3/2}}{24 \sqrt{2\pi}} \Big ( 6 + 4\sigma
  \sqrt{t}d_- + (12 \rho^2 + 4)(d_-^2 -1) + 3\rho^2 \sigma
  \sqrt{t}(d_-^3 - 3d_-) + 3\rho^2(d_-^4 - 6d_-^2 + 3) \Big ),
 \end{equation}
 which follows from Equations \eqref{eq.final.B1} and
 \eqref{eq.final.B2} and Theorem \ref{theorem.main1}.
\end{remark}

\section{Applications: approximations of the derivatives
  and of the implied volatility}
\label{sec.applications}

We now present two applications of the methods that we have
developed.

\subsection{Approximation of derivatives}
Our method is especially well suited for approximating the derivatives
of $C_{SA}$ (some of these derivatives are called ``Greeks'').
Indeed, the derivatives with respect to parameters (other than $\nu$)
of the solution $C_{SA}$ of Equation \eqref{eq.lambdaSABR} and Taylor
expansions in $\nu$ are seen to commute. We shall look in what follows
at $\pa_S C_{SA}$. The asymptotic expansion in $\nu$ for $\pa_S
C_{SA}$ thus can be obtained from the corresponding expansion of
$C_{SA}$ by taking the derivative term by term with respect to $S$ in
Equation \eqref{eq.taylor}.  Let us, for example, compute the
approximation $\Delta_{2} := \pa_S C_{SA,2}(S, K, \nu, \sigma,
t)$. This approximation can be obtained directly by differentiating in
Equation \eqref{eq.option.price1} with respect to $S = e^{x -rt}$ and
$\sigma$. When differentiating with respect to $S$, we also take into
account the fact that $\pa_S = S^{-1} \pa_x$, $\pa_S^2 = S^{-2}
(\pa_x^2 - \pa_x)$, and hence that $\pa_S$ and $\pa_S^2$ commute with
all the other differential operators (this is because all the other
differential operators are polynomials in $\sigma$, $\pa_\sigma$, and
$\pa_x$). We continue to assume that $\kappa_0 = 0$.

For instance, Equation \eqref{eq.final.B1} gives
%\begin{equation}\label{eq.psione'} \pa_x \tilde G_1 \phi_{t}(x) \seq
%   \frac12 \rho \sigma^3 t^2 \tilde H_2 (x) \, \phi_t(x, \sigma) \seq
%   \frac12 \rho \sigma^3 t^2 \big (\frac{1}{\sigma \sqrt{t}} \big )^2
%   H_2 ( \ell(x) ) \, \phi_t(x, \sigma)\,, \end{equation} and hence
\begin{equation}\label{eq.first.correction'}
  \pa_S B_1 h(t, x, \sigma) \seq
%  S^{-1} \pa_x J_1 h(x) =
  \frac{K}{S} \pa_x \tilde G_1 \phi_{t}(x - \ln K) = \frac{K \rho
    \sqrt{t}}{2 S \sqrt{2\pi}} H_2 ( d_- ) e^{-d_-^2/2} \, .
\end{equation}
%We obtain ``at the money'' (that is, for $x = \ln K$) that $\pa_S J_1
%h(\ln K) = \sqrt{t} \vartheta_3(t)$ with $\vartheta_3$ an entire
%function.
Similarly,
\begin{multline}\label{eq.psitwo'}
	\pa_x \tilde G_2 \phi_{t}(x) \seq \pa_x \Big [\,
          \frac{t^2\sigma^2}{4} + \frac{t^3 \rho^2 \sigma^4}{2}
          \pa_x^2 + \frac{t^3\sigma^4}{6} (\pa_x^2 - \pa_x) +
          \frac{t^4\rho^2\sigma^6}{8} (\pa_x^4 - \pa_x^3) \, \Big ]\,
        \phi_t \\
   = \, \Big [\, \frac{t^2\sigma^2}{4} \tilde H_1(x)
     -\frac{t^3\sigma^4}{6} \tilde H_2(x)
	+ \frac{\sigma^4 t^3 (3\rho^2 + 1 )}{6} \tilde H_3(x)
          + \frac{t^4\rho^2\sigma^6}{8} ( \tilde H_5(x) - \tilde
          H_5(x)) \, \Big ]\, \phi_t(x, \sigma) \,,
\end{multline}
and hence
\begin{equation}\label{eq.second.correction'}
	\pa_x B_2 h(t, x, \sigma) \seq K \pa_x \tilde G_2 \phi_{t}(x -
        \ln K) \,.
\end{equation}

Using also that $\pa_S C_{BS} = N(d_+) $, we then obtain the following
theorem, which is an analog for the ``hedging parameter'' $\Delta :=
\pa_S C_{SA}$ of Theorem \ref{theorem.main1}.

\begin{theorem}\label{theorem.main2}
Let $t$ be the time to expiry and $e^x = S e^{rt}$. Then the formal
second order approximation of the $\pa_S C_{SA}(S, K, \nu, \sigma,
t)$, denoted $\pa_S C_{SA,2}(S, K, \nu, \sigma, t)$, is given by
\begin{equation}\label{eq.option.price2}
  	\pa_S C_{SA,2}(S, K, \nu, \sigma, \rho, t) \, = \, N(d_+) +
        \nu e^{-x-rt} ( \pa_x B_1 h + \nu \pa_x B_2 h ) \, ,
\end{equation}
where $\pa_x B_1h$ and $\pa_x B_2h$ are given explicitly by Equations
\eqref{eq.first.correction'} and \eqref{eq.second.correction'}.
\end{theorem}

In other words, we can just compute the derivatives in the
Duhamel-Dyson series term by term.

%%%%%%%%%%%%%%%%%%%%%%%%%%%%%%%%%%%%%%%%%%%%%%%%%%%%%%%%%%%

\subsection{Implied Volatility}
To gain a qualitative insight into the behavior of the implied
volatility in our approximation of the SABR model, it is necessary to
convert \eqref{eq.option.price1} into an asymptotic expansion of the
implied volatility.  We thus assume that the second order
approximation for the implied volatility is of the form
\begin{equation}\label{vol.expansion1}
	\sigma_{imp} \, = \, \sigma+\nu e_1+ \nu^2 e_2 + \, O(\nu^3)
\end{equation}
with the coefficients $e_1$ and $e_2$ to be calculated. We then let $P
(\sigma) := C_{BS}(S, K, \sigma, t)$ and define $\sigma_{imp}$ by the
equation
\begin{equation}\label{eq.matching}
	 P (\sigma_{imp}) \, := \, C_{BS}(S, K, \sigma_{imp}, t) \, =
         \, C_{SA}(S, K, \nu, \sigma, \rho, t) \, .
\end{equation}
We then substitute the assumed form of the $\sigma_{imp}$ function
into the Black Scholes formula $C_{BS}$ and Taylor expand in $\nu$
this composite function around the initial volatility $\sigma$. See
also \cite{Gatheral12a, Gatheral12b, Pascucci12, Pascucci17}. That is,
we have
\begin{multline}\label{eq.option.price3}
	P(\sigma_{imp}) + O(\nu^3) = P(\sigma) + \pa_\sigma
        P(\sigma)(\nu e_1+ \nu^2 e_2) +\frac{\pa^2_\sigma P
          (\sigma)}{2} (\nu e_1+ \nu^2 e_2)^2 \\
	 \, = \, P(\sigma) + \pa_\sigma P(\sigma) e_1 \nu + \, \Big
         (\, \pa_\sigma P(\sigma) e_2 +\frac{\pa^2_\sigma P
           (\sigma)}{2} e_1^2 \, \Big )\, \nu^2\,
	= \, P(\sigma) + \nu e^{-rt} ( B_1 h + \nu B_2 h ) \, .
\end{multline}	
Recall that $F = e^x = S e^{rt}$ and $y = \ln (F/K) = x - \ln K$ is
the log-moneyness. Let $z = d_{-} = \ell(y)$ and $v = \sigma
\sqrt{t}$.  We shall need also the first two derivatives of $P :=
C_{BS}$ with respect to $\sigma$, so we record them here (note that $F
N'(d_+) = K N'(d_-)$)
\begin{equation}\label{eq.deriv.BS}
  	\frac{\pa P}{\pa \sigma} \seq K e^{-rt} \sqrt{t} N'(d_{-}) > 0
        \quad \mbox{ and } \quad
 	\frac{\pa^2 P}{(\pa \sigma)^2} \seq \frac{K e^{-rt} d_{+}d_{-}
          \sqrt{t}}{\sigma} N'(d_{-}) \,,
\end{equation}
which can be made more explicit by substituting $N'(d_{-}) =
\frac{1}{\sqrt{2 \pi}} e^{-d_{-}^2/2}$.  Comparing Equations
\eqref{eq.option.price1} and \eqref{eq.option.price3} and matching the
$\nu$-coefficients according to Equation \eqref{eq.matching}, we
obtain
\begin{equation}\label{vol.coefficient1}
	e_1 \ = \ e^{-rt} J_1 h/\pa_\sigma P(\sigma) = \frac{\rho}{4}
        \big ( \sigma^2 t - 2y\big ) = - \rho v z/2 \,.
\end{equation}
Similarly, let us denote 
\begin{equation*}
	A = 6 + 4 \sigma \sqrt{t} H_1(z) + ( 12 \rho^2 + 4 ) H_2(z)
           + 3 \rho^2\sigma \sqrt{t} H_3(z) + 3 \rho^2 H_4(z) \,,
\end{equation*}
the long factor appearing in the formula for $B_2h$, Equation
\eqref{eq.final.B2}. Then, using $\pa_\sigma P(\sigma) = K e^{-rt}
\sqrt{t} N'(d_{-})$ and $ \sqrt{2\pi} N'(d_{-}) = e^{-\frac{(y
    -\sigma^2 t/2)^2}{2 \sigma^2 t }}$, we obtain
\begin{multline}\label{vol.coefficient2}
	e_2 \ = \ \Big ( e^{-rt} B_2 h-\frac{\pa^2_\sigma P
          (\sigma)}{2} e_1^2 \Big )/\pa_\sigma P(\sigma) \seq
        \frac{e^{-rt} K N'(d_{-}) \sqrt{t}}{ \pa_\sigma P(\sigma) }
        \Big ( \frac{ \sigma t A }{24} \, - \frac{ d_{+}d_{-} e_1^2
        }{2 \sigma } \Big ) \\
	= \frac{ v \sqrt{t} }{ 24 } \big ( A - 6 \rho^2 z^3(z + v)
        \big ) \\
	\seq \frac{\sigma t}{12} - \frac{\rho^2 t\sigma}{8} -
        \frac{\sigma^3t^2}{24} -\frac{\rho^2t\sigma y}{8}+
        \frac{y^2}{6\sigma} - \frac{\rho^2y^2}{4
          \sigma}+\frac{t^2\rho^2\sigma^3}{8} \,.
\end{multline}

We have the following result:

\begin{theorem}\label{theorem.main3}
Let $t$ be the time to expiry and $y=\ln(F/K)=\ln(Se^{rt}/K)$. Then
the implied volatility $\sigma_{imp}$ has an asymptotic expansion of
the form
\begin{equation*}
	\sigma_{imp}( y, \nu, \sigma, \rho, t) = \sigma + \nu e_1 +
        \nu^2 e_2 + \ldots + \nu^k e_k + O(\nu^{k+1}) \,.
\end{equation*}
The coefficients $e_1$ and $e_2$ are given by $e_1 = - \rho \sigma
\sqrt{t} d_{-}/2$ and
\begin{equation*}
	e_2 \seq \frac{\sigma t}{12} - \frac{\rho^2 t\sigma}{8} -
        \frac{\sigma^3t^2}{24} -\frac{\rho^2t\sigma y}{8}+
        \frac{y^2}{6\sigma} - \frac{\rho^2y^2}{4
          \sigma}+\frac{t^2\rho^2\sigma^3}{8} \,.
\end{equation*}
\end{theorem}

%%%%%%%%%%%%%%%%%%%%%%%%%%%%%%%%%%%%%%%%%%%%%%%%%%%%%%%%%

%%%%%%%%%%%%%%%%%%%%%%%%%%%%%%%%%%%%%%%%%%%%%%%%%%%%%%%%%

\section{Model calibration and market tests}
\label{sec.market}

Before we discuss our numerical tests, it is useful to see how our
formulas perform when using market data. The primary reason for
looking at market data is not to show that our method is good
(although we do achieve this, at least partially), but rather to find
the most interesting set of data for which to test numerically our
results in the next section. We thus apply our method to calibrate our
model on some specific data described next and compare the results to
those obtained by using Hagan's approximation recalled below.

\subsection{Description of the data}
The data that we use are options on the SP500 index for 2517
consecutive trading days. Each day has $260 = 2 \times 10 \times 13$
data points: 2 for the choice Call/Put, 10 for the choice of expiry
time $T$:
\begin{equation}
 12\, T \, \in\, \{\ 1,\ 2,\ 3,\ 4,\ 5, \ 6,\ 9,\ 12,\ 18, \ 24
 \ \}\,, 
\end{equation}
and 13 for the choices of $\Delta$, namely,
\begin{equation}
   \Delta \, \in\, \{\ 0.2,\ 0.25,\ 0.3,\ \ldots,\ 0.75,\ 0.8\ \}
\end{equation}
(the value $y = \ln(S e^{rt}/K)$ is determined from $\Delta$ as
explained below, Remark \ref{rem.Delta}).

\subsection{General description of the method}
We run {\em three types of tests} using the market data. Each type of
tests is distinguished by its {\em objective function.} More
precisely, for the first type of tests we use the implied volatility
as an objective function, for the second type of tests we use the
actual price (obtained from the implied volatility and the
Black-Scholes formula) as an objective function, and, finally, for the
third type of tests we use the logarithm of the actual price.  This is
explained in more detail next.

For each of the three objective functions and each of the 2517 trading
days, we use the least squares optimization to estimate the model
parameters $(\nu, \sigma, \rho)$ to best fit the market data. (See
Equation \eqref{eq.er.one} for an explicit formula.)  We list the
fitting error (which is an ``in-sample'' error) in the first column of
Table \ref{table.1}. We also use the parameters $(\nu, \sigma, \rho)$
to predict the market data for the next day, and then we list the
corresponding (``out-of-sample'') error. We looked at different types
of ``errors'' (more precisely, norms): $\ell_1$-, $\ell_2$-, and
$\ell_\infty$-norms, but for simplicity we will mostly discuss the
$\ell_2$-norm error. All these norms are normalized, in the sense that
we always use averages (or probability measures). More specifically,
for a sequence $a := (a_i)_{i \in I}$, where $I$ is a finite set with
$|I|$ elements, its $\ell_2$-norm $\|a\|_{2}$ is given by
\begin{equation}
   \|a\|_{2}^2 \ede |I|^{-1}\, \sum_{i \in I}\, a_i^2 \,.
\end{equation}
(When complex numbers are used, $a_i^2$ is replaced with $|a_i|^2$,
but that will not be necessary in what follows.) Similarly, the
$\ell_1$- and the $\ell_\infty$-norms of $a := (a_i)_{i \in I}$ are
given by $\|a\|_{1} := |I|^{-1}\sum_{i \in I} |a_i|$ and by
$\|a\|_{\infty} := \max_{i \in I} |a_i|$, respectively. For the
$\ell_\infty$-norm, the normalization plays, of course, no role.

Thus, for each trading day $\tau = 1, \ldots, N$ we have a set of $J =
260$ triples $\maM_{\tau} = \{(y_{\tau, j}, \Sigma_{\tau, j}, T_{\tau,
  j})\}$, (so $j = 1, \ldots, J$). Each triple in the above set
represents the moneyness, the quoted volatility, and the time to
expiry, respectively.  Assuming that the strike is always $K = 1$ and
$r = 0$, which does not change our calculations in view of Remarks
\ref{rem.homogeneous} and \ref{rem.RelBS}, the price of the option for
this trade is then obtained using the function $C_{rel}$ of Equation
\eqref{eq.def.C_rel}
\begin{equation}
 p_{\tau, j} \, = \, C_{rel}(y_{\tau,j}, \Sigma_{\tau,j},
 T_{\tau,j})\, .
\end{equation}

More precisely, the procedure that we have outlined in the beginning
of this subsection amounts to the following.  For each objective
function and each trading day $\tau$, we use the least squares
optimization to find the the parameters $(\nu_\tau, \sigma_\tau,
\rho_\tau)$ for which the resulting objective function best matches
(in an $\ell_2$ sense) the quoted data.

Before explaining this in even more detail for each of the three types
of tests, let us make first the following remark about the structure of
our data. 

\begin{remark}\label{rem.Delta}
In our data, the market data does not provide the log-moneyness value
$y_{\tau, j}$, but rather the hedging parameter $\Delta_{\tau,
  j}$. The value $y = ln(S e^{rt}/K)$ is then estimated from $\Delta$
using the formula
\begin{equation}
  y_{\tau, j} \ =\ \frac{\sigma_{\tau-1} \sqrt{T_{\tau, j}}}{2} \,
  \big(\, 2 N^{-1}( \Delta_{\tau, j}) - \sigma_{\tau-1} \sqrt{T_{\tau,
      j}} \, \big),
\end{equation}
where $\sigma_{\tau-1}$ is the implied volatility determined in the
previous trading day and $T = t$. For the first trading day we took
for the parameters as an initial value the average value from a
previous run of the program. We also estimated $y_{\tau, j}$ using
$\Sigma_{\tau, j}$ instead of $\sigma_{\tau-1}$, but did not find any
significant differences in the error estimates reported in Table
\ref{table.1}.
\end{remark}

\subsection{The first type of market data tests: implied volatility}
For the first type of tests, we have used two types of implied
volatility approximations as an objective function: $\sigma_D$ and
$\sigma_H$ defined below.

\subsubsection{Objective function: $\sigma_D$}
The first choice of implied volatility function, namely $\sigma_D$, is
the one provided by Theorem \ref{theorem.main3} by truncating to the
second order approximation in powers of $\nu$, that is,
\begin{equation}\label{our.impvol}
   \sigma_{D}(y, \nu, \sigma, t) \ede \sigma + \nu e_1 + \nu^2 e_2\,.
\end{equation}
We used this objective function as follows.

Recall that $J_\tau$ is the set of market data for the day $\tau$.
For each trading day $\tau$, we have computed the parameters
$(\nu_{\tau}, \sigma_{\tau}, \rho_{\tau})$ that minimize the (square
of the) $\ell_2$-error between our objective function ($\sigma_D$) and
the market provided analog function (the implied volatility, denoted
$\Sigma$ decorated with various indices in this case). That is, we
have chosen $(\nu_{\tau}, \sigma_{\tau}, \rho_{\tau})$ to minimize
\begin{multline}\label{eq.er.one}
 \| \sigma_D (J_\tau, \nu, \sigma, \rho) - \Sigma(J_\tau) \|_2^2 \ede
 \frac{1}{|J_\tau|} \sum_{(y, \Sigma, T) \in J_\tau} \big(\,
 \sigma_D(y, \nu, \sigma, \rho, T) - \Sigma\, \big )^2 \\
 = \, \frac{1}{260} \sum_{j = 1}^{260} \big(\, \sigma_D(y_{\tau, j},
 \nu, \sigma, \rho, T_{\tau, j}) - \Sigma_{\tau, j}\, \big )^2 \,.
\end{multline}

The resulting minimum value for the $\ell_2$ norm is the fitting error
of our model, or the ``in-sample-error'' (ISE), obtained by replacing
$(\nu, \sigma, \rho)$ with $(\nu_{\tau}, \sigma_{\tau}, \rho_{\tau})$:
\begin{equation}\label{eq.er.two}
 \| \sigma_D (J_\tau, \nu_{\tau}, \sigma_{\tau}, \rho_{\tau}) -
 \Sigma(J_\tau) \|_2^2 \ede \frac{1}{260}\, \sum_{j = 1}^{260}\,
 \big(\, \sigma_D(y_{\tau, j}, \nu_{\tau}, \sigma_{\tau}, \rho_{\tau},
 T_{\tau, j}) - \Sigma_{\tau, j}\, \big )^2 \,.
\end{equation}
We have also computed the ``out-of-sample'' (OSE) error obtained by
replacing $(\nu, \sigma, \rho)$ with $(\nu_{\tau-1}, \sigma_{\tau-1},
\rho_{\tau-1})$, that is, with the parameters obtained in the previous
day.
\begin{equation}\label{eq.er.three}
 \| \sigma_D (J_\tau, \nu_{\tau-1}, \sigma_{\tau-1}, \rho_{\tau-1}) -
 \Sigma(J_\tau) \|_2^2 \ede \frac{1}{260}\, \sum_{j = 1}^{260}\,
 \big(\, \sigma_D(y_{\tau, j}, \nu_{\tau-1}, \sigma_{\tau-1},
 \rho_{\tau-1}, T_{\tau, j}) - \Sigma_{\tau, j}\, \big )^2 \,.
\end{equation}
The second line of the following table (Table \ref{table.1})
summarizes the results of our test, by providing the ISE and OSE, the
average values $\overline{\nu}$, $\overline{\sigma}$,
$\overline{\rho}$, as well as their standard deviations $std(\nu)$,
$std(\sigma)$, and $std(\rho)$ for the specified objective function
($\sigma_D$ in this case).
\medskip

%\newpage
\begin{table}
\caption{Comparison of the implied volatilities of the models with the
  market data}
\label{table.1}
\begin{center}
\begin{tabular}{ | c | c | c | c | c | c | c | c | c | }
 \hline
 Obj. & ISE & OSE & $\overline{\nu}$ & $std(\nu)$ &
 $\overline{\sigma}$ & $std(\sigma)$ & $\overline{\rho}$ & $std(\rho)$
 \\
 \hline 
 $\sigma_D$ & \ .0152 \ & \ .0179 \ & \ 1.335 \ & \ 0.3 \ & \ .1889
 \ & \ .078 \ & \ -0.54 \ & \ .07 \ \\
 \hline
 $\sigma_H$ & \ .0154 \ & \ .0181 \ & \ 1.332 \ & \ 0.29 \ & \ .1902
 \ & \ .079 \ & \ -0.58 \ & \ .069 \ \\
 \hline
\end{tabular}
\end{center}
\end{table}
\medskip

\subsubsection{Objective function: $\sigma_H$}
Recall that in \cite{Les03}, Hagan, Kumar, Lesniewski, and Woodward
applied a singular perturbation technique to derive a closed form
implied volatility approximation under SABR dynamics, denoted
$\sigma_H$. Let $z=\frac{\nu}{\sigma}\ln(F/K)$ and
\begin{equation}\label{eq.def.xi}
	\xi(z) \, = \, \ln{\frac{\sqrt{1-2\rho z+z^2}+z-\rho}{1-\rho}}
        \,.
\end{equation}
For $\beta=1$, the only case considered in this paper, their implied
volatility approximation is given by the formula:
\begin{equation}\label{Hagan implied vol}
	\sigma_{H}(y, \nu, \sigma, \rho, t) \ede \sigma
        \frac{z}{\xi(z)} \Big [1 + \big (\frac{1}{4}\rho\nu\sigma +
          \frac{2-3\rho^2}{24}\nu^2 \big )t \Big ] \,,
\end{equation}
(See also \cite{Antonov15}.)
The last data line of at table, Table \ref{table.1}, summarizes the
same information as in the previous paragraph, but
for $\sigma_H$ replacing $\sigma_D$. 

\subsubsection{Conclusion}
The calculations summarized in Table \ref{table.1} show that the two
formulas, $\sigma_H$ and $\sigma_D$, yield {\em very similar results
  for the parameters appearing in the data set} (this is not true in
general!).

\subsection{The second type of market data tests: actual prices} 
In this subsection we report the results of the tests that compared
the actual prices The choice of an approximation of the implied
volatility yields an approximation of the price through the
Black-Scholes formula \eqref{eq.BS}. For instance, if we approximate
$\sigma_{imp}$ with $\sigma_D$ (as in the first set of tests in the
previous subsection), the resulting approximation for the price is
given by
\begin{equation}\label{eq.our.imp.price}
  C_D(y, \sigma, t) \ede C_{rel}(y, \sigma_{D}, t) \,.
\end{equation}
(In particular, $C_D = e^{-rt}F_D$, see Equation \eqref{eq.imp.2}.)
Analogously, if we approximate $\sigma_{imp}$ with $\sigma_H$, the
resulting formula will be denoted
\begin{equation}\label{eq.Hagan.price}
	C_H(y, \sigma, t) \ede C_{rel}(y, \sigma_{H}, t) \,.
\end{equation}
The approximation $C_{H}$ of the solution of the SABR PDE is widely
considered as very accurate for short dated options. We have compared
the resulting values $C_H$ and $C_D$ with the prices $C_M$ provided by
the market data. In addition to these approximate price functions, we
have also tested a regularization $\tilde C_H$ of Hagan's formula, as
well as our second order approximation of $C_{SA}$, that is $C_{SA, 2}
= e^{-rt}F_{SA, 2}$ of Theorem \ref{theorem.main1}. As explained in
Remarks \ref{rem.homogeneous} and \ref{rem.RelBS}, we can assume $K =
1$ and $r = 0$, up to multiplying all the results with a global
factor.

\subsubsection{Objective functions: $C_D$ and $C_{SA, 2}$}
We have then proceeded as in the previous subsection, but we have used
$C_D$ instead of $\sigma_D$ and the actual prices $C_M = C_{BS}(S, K,
\Sigma, t)$ instead of $\Sigma$, where $\Sigma$ is the implied
volatility and is provided by the data set. We have also included a
test for the formula for the approximate price including the mean
reverting term for $\kappa = .25$. In our tests, we have assumed
$Ke^{-rt} = 1$. Instead of \eqref{eq.er.one}, we have thus chosen
$(\nu, \sigma, \rho)$ to minimize
\begin{equation}\label{eq.er.one.prime}
 \sum_{j = 1}^{260}\, \big(\, C_{rel}(y_{\tau, j},
 \sigma_D(y_{\tau, j}, \nu, \sigma, \rho, T_{\tau, j}), T_{\tau, j}) -
 C_{rel}(y_{\tau, j}, \Sigma_{\tau, j}, T_{\tau, j}) \, \big )^2 \,,
\end{equation}
where $C_{rel}$ is as defined in Equation \eqref{eq.def.C_rel} (we
have thus replaced $\sigma_D(y_{\tau, j}, \nu, \sigma, \rho, T_{\tau,
  j}) - \Sigma_{\tau, j}$ with $C_{rel}(y_{\tau, j}, \sigma_D(y_{\tau,
  j}, \nu, \sigma, \rho, T_{\tau, j}), T_{\tau, j}) - C_{rel}(y_{\tau,
  j}, \Sigma_{\tau, j}, T_{\tau, j})$ in Equation
\eqref{eq.er.one}). To obtain the In Sample Error (ISE) and the Out of
Sample Error (OSE), we have performed some similar modifications to
Equations \eqref{eq.er.two} and \eqref{eq.er.three}. That has amounted
to replacing $(\nu, \sigma, \rho)$ in Equation \eqref{eq.er.one.prime}
with $(\nu_{\tau}, \sigma_{\tau}, \rho_{\tau})$ in order to obtain the
ISE and with $(\nu_{\tau-1}, \sigma_{\tau-1}, \rho_{\tau-1})$ in order
to obtain the OSE. We proceeded similarly for $C_{SA, 2} = e^{-rt}
F_{SA, 2}$ (see Equation \eqref{eq.SABR.2}). The results are provided
in the last two columns of Table \ref{table.2}.

\subsubsection{Objective functions: $C_H$ and $\tilde C_H$}
The similar results for $C_D$ replaced with $C_H$ are provided in the
second column of that table. It turns out that there are some
numerical instabilities in the formula for $\sigma_H$ due to the
quotient $z/\xi(z)$, since $\xi(0) = 0$, so for very small $z$ we get
into issues of machine precision.  For this reason, for $z$ very
small, we interpolate between $z/\xi(z)$ and $1 - \frac{\sigma z}{2}
\approx z/\xi(z)$ to obtain a new formula $\tilde C_H$, for which we
also summarize the (slightly better) results.  When dealing with
market data, $z$ is, in fact, never zero, but when dealing with
numerical tests, we do get $z = 0$, so this interpolation becomes
indispensable.

\medskip

\begin{table}
\caption{Comparison of the actual prices of the models with the market
  data}
\label{table.2}
\begin{center}
\begin{tabular}{ | c | c | c | c | c | c | c | c | c | }
 \hline
 Obj. & ISE & OSE & $\overline{\nu}$ & $std(\nu)$ &
 $\overline{\sigma}$ & $std(\sigma)$ & $\overline{\rho}$ & $std(\rho)$
 \\
 \hline 
 $C_H$ & \ .0038 \ & \ .0043 \ & \ 1.0463 \ & \ .24 \ & \ .197 \ &
 \ .07 \ & \ -0.56 \ & \ .07 \ \\
 \hline
 $\tilde C_H$ & \ .0036 \ & \ .0042 \ & \ 1.0687 \ & \ .28 \ & \ .192
 \ & \ .07 \ & \ -0.39 \ & \ .15 \ \\
 \hline
 $C_D$ & \ .0037 \ & \ .0042 \ & \ 1.0502 \ & \ .25 \ & \ .196 \ &
 \ .07 \ & \ -0.53 \ & \ .08 \ \\
 \hline
 $C_{SA,2}$ & \ .0037 \ & \ .0043 \ & \ 1.0566 \ & \ .25 \ & \ .196
 \ & \ .07 \ & \ -0.55 \ & \ .07 \ \\
 \hline
\end{tabular}
\end{center}
\end{table}
\medskip

\subsubsection{Conclusion}
Replacing $C_H$ with $\tilde C_H$ improves the performance of this
formula by a small, but significant amount. In
any case, all these models (except the Black-Scholes model) perform in
a very similar way, as seen by examining the In Sample Errors (ISE)
and the Out of Sample Errors (OSE).

\subsection{The third type of market data tests: log-prices}
\label{ssec.w.kappa}
The reader may have already observed the fallacy of the method used in
the previous section: by looking at the differences $C_M - C_D$, for
example, we do not take account the fact that we want a {\em smaller
  error} when $C_M$ is small.  Because of this, we now perform the
same tests, but for the logarithm of the prices.  We thus replace
$C_D$ with $\ln(C_D)$ (and similarly for the other $C$s). Thus, for
the objective function $\ln(C_D)$, we minimize
\begin{equation}\label{eq.er.one.log}
   \sum_{j = 1}^{260}\, \big[\, \ln \big(C_{rel}(y_{\tau, j},
     \sigma_D(y_{\tau, j}, \nu, \sigma, \rho, T_{\tau, j}), T_{\tau,
       j}) \big) - \ln \big( C_{rel}(y_{\tau, j}, \Sigma_{\tau, j},
     T_{\tau, j}) \big) \, \big ]^2 \,,
\end{equation}
All the other formulas for the other objective functions change in a
similar way. We have also tested the mean reverting term. We obtain
the results summarized in Table \ref{table.3}, with the mean reverting
term in the row corresponding to $\ln ( C_{\kappa} )$.

\begin{table}
\medskip
\caption{Comparison of the log prices of the models with the market
  data}
\label{table.3}
\begin{center}
\begin{tabular}{ | c | c | c | c | c | c | c | c | c | }
 \hline
 Obj Funct & ISE & OSE & $\overline{\nu}$ & $std(\nu)$ &
 $\overline{\sigma}$ & $std(\sigma)$ & $\overline{\rho}$ & $std(\rho)$
 \\
 \hline 
 $\ln (C_H)$ & \ .069 \ & \ .077 \ & \ 1.2 \ & \ .2 \ & \ .19 \ &
 \ .08 \ & \ -0.57 \ & \ .06\ \\
 \hline
 $\ln (\tilde C_H)$ & \ .056 \ & \ .066 \ & \ 1.4 \ & \ .4 \ & \ .18
 \ & \ .08 \ & \ -0.28 \ & \ .12 \ \\
 \hline
 $\ln (C_D)$ & \ .059 \ & \ .068 \ & \ 1.5 \ & \ .3 \ & \ .19 \ &
 \ .08 \ & \ -0.54 \ & \ .04 \ \\
 \hline
 $\ln (C_{SA, 2})$ & \ .067 \ & \ .076 \ & \ 1.3 \ & \ .3 \ & \ .19
 \ & \ .08 \ & \ -0.55 \ & \ .05 \ \\
 \hline
 $\ln (C_{\kappa})$ & \ .064 \ & \ .074 \ & \ 1.2 \ & \ .3 \ & \ .19
 \ & \ .08 \ & \ -0.55 \ & \ .04 \
% }
\\ 
 \hline
 $\ln (C_{BS})$ & \ .119 \ & \ .123 \ & \ \ & \ \ & \ .18 \ & \ .06
 \ & \ \ & \ \ \\
 \hline
\end{tabular}
\end{center}
\medskip
\end{table}

Here are some comments on these results. The best model is the
modified Hagan model, followed closely by our implied volatility model
(so it is an advantage in this case to consider rather the implied
volatility formula of Theorem \ref{theorem.main3} instead of the
approximate price formula of Theorem \ref{theorem.main1}, although
they differ by a term of order $O(\nu^3)$). For comparison, on the
last line, we also show the performance of the Black-Scholes model,
which is seen to be significantly worse than all the other ones.
Using an expansion in $\nu$ of the implied volatility improves the
results compared to the similar approximation of the
price. Introducing the mean reverting term also improves (for small
$\kappa$) the results. This is seen by comparing the row corresponding
to $\ln (C_{\kappa})$ with the row corresponding to $\ln (C_{SA, 2})$,
with the formula for $C_{\kappa}$ being obtained from the formula for
$C_{SA, 2}$ by including the mean reverting term.

\begin{remark}
We have included the information on the average values and the
standard deviations of the parameters $(\nu, \sigma, \rho)$ since we
will use them to select the most relevant parameters for our numerical
tests. In this sense, we notice that $y$ has mean $-.01$ and standard
deviation $0.28$.
\end{remark}

\section{Numerical tests}\label{sec.num}

We have tested numerically our approximation formulas $C_{SA, 2}$ and
$C_D$ in several ways. The numerical tests confirm the efficiency of
our method in a suitable range of the parameters.  In all our
numerical tests below, we have chosen $K = 1$ and $\kappa = 0$. We
have also chosen $r = 0$, so the distinction between the forward
prices $F = e^{rt}C$ and the actual prices $C$ disappears.

\subsection{The residual of the approximations: substituting in the
PDE}\label{sec.diff.eq} The simplest numerical test to estimate the
performance of the various methods considered in this paper is to
check if the resulting price function $C$ satisfy the SABR PDE
(Equation \eqref{eq.lambdaSABR}). Recall that $L$ is the generator of
the SABR PDE and $\kappa = 0$ in our numerical tests. That is, we
compute the residual $\pa_t C - LC$, where $C = C_H$ for Hagan's
model, $C = C_D$ for our model using the approximate volatility
$\sigma_D$, and $C = C_{SA, 2}$ for our second order approximation. If
$C$ was an exact solution of the SABR PDE, the residual would be
zero. So the norm $R = \|\pa_t C - LC\|$ of the residual gives us an
idea how far $C$ is from the actual solution.  For comparison, we
compute also the norm of the residual when $C = C_{BS}$, that is, when
$C$ is given simply by the Black-Scholes formula. One can say that if
the norm of the residual is small, then $C$ is close to the actual
solution, since the PDE is well-posed for $\kappa = 0$. The converse
is not true, however, in that a large residual does not imply that $C$
is far from the solution (this is the same phenomenon as the one that
gives that $f$ is small if $f'$ is small, but not the other way
around).

We computed the $\ell_2$-norm of the residual over various
regions. For instance, let us consider the set of data points for
which $\nu = 0.125$, $T$ belongs to a set of equally spaced nodes
between $0.1$ and $1$, $\rho = -0.4$, $\sigma$ belongs to a set of
equally spaced nodes between $0.1$ and $0.3$, and $y$ belongs to a set
of equally spaced nodes between $-0.5$ and $0.5$.  The norms of the
residuals for our choices of $C$ are summarized in Table
\ref{table.4}. (In that table, we showed the norms times 1000, since
the numbers were otherwise very small.)

\begin{table}
\medskip
\caption{Estimate of the residuals $R:=\| (\pa_t - L) C \|_{2}$}
\label{table.4}
\begin{center}
\begin{tabular}{ | c | c | c | c | c |  }
 \hline
 $C =$  \ & $C_H$ & $C_D$ & $C_{SA, 2}$ & $C_{BS}$ \\
 \hline 
 $10^3 R $ \ & \ 0.489 \ & \ 0.181 \ &
 \ 0.163 \ & \ 16.436 \\ \hline
\end{tabular}
\medskip

$\nu = .125$,\ $T \in [0.1, 1]$,\ $\rho = -.4$,\ $\sigma \in [0.1,
  0.3]$,\ and $y \in [-0.5, 0.5]$.
\end{center}\medskip
\end{table}

We performed the same test for various other regions, including for
regions in which the values of the parameters were close to the market
data. The results, as expected, get worse with the increase of $\nu$,
of $T$, and of the size of the interval for $y$.  In these tests, our
model behaves better for small $\nu$, even for large $T$. However, for
larger values of $\nu$, Hagan's model behaves better (for this type of
tests).

In all these tests, we have used only values of $T\ge 0.1$, since the
numerical differentiation used in the program becomes less reliable
for $T < 0.1$.

\begin{table}
\caption{Estimates of residuals (continued)}
\label{table.5}
\begin{center}
\begin{tabular}{ | c | c | c | c | c | c | c | c |  }
 \hline
 $C =$ \ & $C_H$ & $C_D$ & $C_{SA,2}$ & $C_{BS}$ &
 $\nu$ & $T \in $ & $y \in $ \\
 \hline 
  $10^2 R $ \  & \ .33 \ & \ .07 \ & \ .072
 \ & \ 1.8 \ & \ .1 \ & \ \ $[\, 0.1,\, 30\, ]$ \ \ & \ \ $[-0.3, .3]$
 \ \\
 \hline 
 $10^2 R $ \  & \ .38 \ & \ .28 \ & \ .37
 \ & \ 1.8 \ & \ .1 \ & \ \ $[\, 0.1,\, 30\, ]$ \ \ & \ \ $[\, -1.5,\,
   1.5\, ]$ \ \\

 \hline 
 $10^2 R $ \  & \ .02 \ & \ .016 \ &
 \ .018 \ & \ 1.4 \ & \ .25 \ & \ \ $[\, 0.1,\, 0.2 \, ]$ \ \ &
 \ \ $[\, -1.5,\, 1.5\, ]$ \ \\
 \hline 
 $10^2 R $ \  & \ 3.2 \ & \ 2.4 \ & \ 5.2
 \ & \ 18.8 \ & \ 1 \ & \ \ $[\, 0.1,\, 1\, ]$ \ \ & \ \ $[\, -0.2,\,
   0.2\, ]$ \ \\
 \hline 
 $10^2 R $ \  & \ 4.3 \ & \ 26.  \ &
 \ 14.8 \ & \ 15.5 \ & \ 1 \ & \ \ $[\, 0.1,\, 1\, ]$ \ \ & \ \ $[\,
   -1,\, 1 \, ]$ \ \\
 \hline 
 $10^2 R $ \  & \ 6.1 \ & \ 4.  \ & \ 7.
 \ & \ 17.5 \ & \ 1 \ & \ \ $[\, 0.1,\, 2\, ]$ \ \ & \ \ $[\, -0.2,\,
   0.2\, ]$ \ \\
 \hline 
\end{tabular}
\end{center}
\color{black}
\end{table}

\medskip

\subsection{Comparison of our implied volatility formula with Hagan's 
formula for maket data} For all the parameters $q_{\tau, j} :=
(\nu_{\tau}, \sigma_{\tau}, \rho_{\tau}, y_{\tau, j}, T_{\tau, j})$
determined as explained in the previous section, we have also computed
the differences between the implied volatilities $\sigma_H -
\sigma_D$, the prices
\begin{multline*}
  C_H(q_{\tau, j}) - C_D(q_{\tau, j}) \ede
  C_{rel}(\sigma_H(q_{\tau,j})) - C_{rel}(\sigma_D(q_{\tau,j}))\\
  \ede C_{rel}( y_{\tau, j}, \sigma_H(\nu_{\tau}, \sigma_{\tau},
  \rho_{\tau}, y_{\tau, j}, T_{\tau, j}), T_{\tau, j}) - C_{rel}(
  y_{\tau, j}, \sigma_D(\nu_{\tau}, \sigma_{\tau}, \rho_{\tau},
  y_{\tau, j}, T_{\tau, j}), T_{\tau, j}).
\end{multline*}
We have similarly computed the differences between the corresponding
log-prices:
\begin{equation*}
  \ln (C_H(q_{\tau, j})) - \ln (C_D(q_{\tau, j})) \ede
  \ln(C_{rel}(\sigma_H(q_{\tau,j}))) -
  \ln(C_{rel}(\sigma_D(q_{\tau,j})))\,.
\end{equation*}
For all these options, we have computed the $\ell_1$, $\ell_2$, and
$\ell_\infty$ norm, results that are summarized in Table
\ref{table.diff.2}

\medskip

\begin{table}
\caption{Comparison of Dyson series method with Hagan's formula}
\label{table.diff.2}
\begin{center}
\begin{tabular}{ | c | c | c | c |  }
 \hline
 Obj Funct & $\ell_1$ & $\ell_2$ & $\ell_\infty$ \\ \hline
 $\sigma_H - \sigma_D$ & \ .0057 \ & \ .0139 \ & \ .167 \ \\
 \hline
$C_H - C_D$ & \ .0010 \ & \ .0026 \ & \ .0266 \ \\
 \hline $\ln (C_H) - \ln (C_D) $ & \ .0326 \ & \ .0938 \ & \ 3.598
 \ \\ \hline
\end{tabular}
\end{center}
\end{table}
\medskip

The conclusion is that, {\em on average,} $\sigma_H$ and $\sigma_D$
predict some very close values. However, occasionally, these values may be
very different.  The ``best results'' are obtained for the differences
in price, but these are not too relevant, since they are not
dimensionless (they do not take into account the magnitude of the
prices).

\subsection{A Monte Carlo simulation}
None of the formulas $C_H$, $C_{D}$, or $C_{SA, 2}$ is an exact solution
of the SABR PDE. In order to compare them to the true solution, we need
a method to approximate very well the true solution, so to see which of
the methods approximates it best. Recall that $\kappa = 0$ in our
numerical test.

In a first set of tests, we have used the Monte Carlo Method to
approximate the true solution of the SABR PDE.  To carry out the
tests, we have first computed prices of call options through Monte
Carlo simulation (denoted below by $C_{MC}$) using 30000 paths and
time step $10^{-4}$. We took thise values as the benchmark.  We have
then calculated the differences $E_{H}=C_{H}-C_{MC}$ and
$E_{D}=C_{D}-C_{MC}$ for a range of $K$ and $t$ and ploted them as a
function of the moneyness $y = \ln{F/K}$. We choose the following
parameters throughout the test: $F=10$, $\sigma=0.2$, $\nu=0.2$,
$\rho=-0.3$.

% As shown in Figure~\ref{fig:error}, t
The approximations $C_H$ and $C_D$ agree almost exactly when $t=1$,
and are both fairly accurate. The largest error is 0.8\%, or around
0.1\% to 0.2\% of the benchmark price $C_{MC}$. As $t$ increase to 3
and 10, the error starts to increase and the two approximations
gradually diverge. When both $C_{H}$ and $C_{D}$ overestimate
$C_{MC}$, Hagan's approximation tends to give a better approximation,
whereas if they both underestimate $C_{MC}$, Duhamel-Dyson
perturbative series approximation has a smaller error. So overall it
is not possible to tell which of the two approximations of the implied
volatility (Hagan's and ours) is better for small $t$. Finally, when
$t$ is as large as 30 years. The difference between $C_{H}$ and
$C_{D}$ becomes significant. It is interesting to see that both
$C_{H}$ and $C_{D}$ overestimate $C_{MC}$ most of the time and,
moreover, $C_{D}$ is almost systematically below $C_{H}$. As a result,
Duhamel-Dyson perturbative series method in this case gives a better
approximation for a majority of strikes. Indeed, the largest relative
error for Hagan's approximation is 22\%, however for Duhamel-Dyson
perturbative series expansion method is only 12\%.

% \begin{figure} \begin{center} \includegraphics[scale =
%0.65]{errorall.pdf} \end{center} \caption{Pricing errors for Hagan
%and Dyson approximations.}\label{fig:error} \medskip This figure
%plots pricing errors for Hagan and Dyson series approximations to an
%option price for $t=1,\, t = 3,\, t = 10,$ and $t = 30$. The
%horizontal axis shows $\ln\frac{F}{K}$, the log of price to strike
%ratio.  \end{figure}

For $\nu = 1.5$, our Monte-Carlo simulation, however, does not
converge even for $10^6$ paths. For this reason we tried also a Finite
Difference simulation.

\subsection{Comparison with the Finite Difference approximate solution}
In a second set of tests, we took as benchmark a Finite Difference
approximation of the solution to the SABR PDE. This is the most
relevant method (so we saved the best for the last), since the FD
method is more precise than the Monte Carlo method and, in our case,
leads to a rather good approximation of the solution (but not
perfect).

In fact, we have computed a sequence $w_k$ of Finite Difference (FD)
approximations of the solution $F_{SA}$ of Equation
\eqref{eq.lambdaSABR} with $\kappa_0 = 0$ obtained by successively
refining our set of nodes. For this sequence $w_k$ of approximations,
we have tested the convergence of the FD approximations and we have
compared them to the various approximations $C$ that we have
considered before: $C_{SA, 2}$, $C_{D}$, $C_{H}$, and $C_{BS}$. The
results of these comparisons for the finest discretization (largest
$k$) are shown in Tables \ref{table.FD1} and \ref{table.FD2}. Let us
now explain the results included in those talbles. We continue to
assume $r = 0$, so $F_{SA} = C_{SA}$, $F_D = C_D$, and so on.

\subsubsection{Outline of the method and of its challenges}
The tests in Section \ref{sec.market} (using market data) have given
for $y$ a range approximately contained in $[-0.3, 0.3]$.  Therefore,
in our FD tests, we have taken $I := [-1, 1]$ as an interval of
interest for $y$. That means that we have compared our FD
approximations $w_k$ with the values $C$ predicted by the other methods
{\em only} for $y \in [-1, 1]$. We similarly took as an interval of
interest for $\sigma$ the interval
\begin{equation*}
  J \ede [0.14 , 0.23] \seq [ 0.18/c, 0.18c \big]
\end{equation*}
(that is, $J$ is an interval symmetric with respect to $0.18$ on a
logarithmic scale; but note that we have rounded off our values to two
significant digits). In view of our results involving the market data
(see Section \ref{sec.market}), the choice of $0.18$ for the average
volatility seems to be a reasonable choice.  We have then compared our
FD approximations $w_k$ with the the other approximations $C$ (namely $C_H$,
$C_V$, $C_{SA, 2}$, and $C_{BS}$) at the nodal points of
\begin{equation}\label{eq.def.Interval}
  (x, \sigma) \in I \times J \seq [-1, 1] \times [\, 0.14 ,\, 0.23\, ] \, .
\end{equation}
The parameters $\nu, \rho, T$ have also been chosen to be close to the
ones provided by the market data (as we will explain below).

In the implementation and error analysis, we had to deal with the fact
that the domain is non-compact and that the coefficients of the SABR
PDE are very far from being constant, in fact, they are very small in
some regions and very large in others. Let us explain next how we
specifically dealt with these issues.

\subsubsection{Cut-off errors}
Recall that the equation that we want to solve is Equation
\eqref{eq.lambdaSABR} (for $\kappa = 0$, in these tests) on the domain
$(x, \sigma) \in \RR \times (0, \infty)$.  Since the domain is
non-compact, in order to discretize this equation using Finite
Differences (FD), we have performed a cut-off of the domain by
restricting to a rectangular domain $\Omega := [ -x_{MAX}, x_{MAX}]
\times [\sigma_{min}, \sigma_{MAX}]$ with $\sigma_{min}\sigma_{MAX} =
0.18^2$ and $\Omega$ significantly larger than $I \times J$. This
required us then to specify the values for the solution on the
boundary $\pa \Omega$ of the domain $\Omega$, something that had not
been needed before the cut-off (since, instead of prescribing boundary
conditions, we had assumed that our solution $u(t)$ is in a suitable
weighted Sobolev space). We have thus replaced our problem
\eqref{eq.lambdaSABR} with
\begin{equation}\label{eq.boxSABR}
\begin{cases}
  \ \partial_{t}w \, = \, \sigma^2 \big [ \frac{1}{2}
    (\partial^2_xw - \partial_xw) + \nu \rho \partial_x\partial_\sigma
    w + \frac{1}{2} \nu^2 \partial^2_\sigma w \, \big ] & \ (x,
  \sigma) \in \Omega \\
  \ w(x, \sigma, 0) \seq |e^x - 1|_{+} \, & \ \mbox{ and } \\
  \ w(x, \sigma, t) \seq g(x, \sigma, t) \, & \ \mbox{ if } (x,
  \sigma) \in \pa \Omega \,.
\end{cases}
\end{equation}
By classical results, see  \cite{CroisilleBook, CroisilleContrast, TyskFD09, 
SchwabFinBook, TyskAdaptive, ObermanDegenerate, 
ObermanAdaptive} for instance, the FD approximations of this equation converge
to the solution $w$ of Equation \eqref{eq.boxSABR}, something that we
have noticed very clearly in our implementation.
  
The solution $w$ of Equation \eqref{eq.boxSABR} is, however, different
from the solution $ C_{SA}$ of Equation \eqref{eq.lambdaSABR} for $v =
h$ and with $\kappa = 0$. Nevertheless, in view of Remark \ref{rem.wp}
and of the exponential decay of the Green function of the equation
$\pa_t - L$, one can prove for $g$ not too far from the actual
solution on the boundary (even $g = 0$ will do) that $\|u-w\|_{I
  \times J} \to 0$ as $\Omega$ approaches the total domain $\RR \times
(0, \infty)$ (see \cite{choulli, MazzucatoNistor1} and the references
therein. See also \cite{TyskBdry11}. Here the norm is the normalized
$\ell_2$ norm defined using the nodal points in the domain of interest
$I \times J$.  Therefore, a large cut-off will not affect too much the
desired values on the small rectangle of interest $I \times J$. The
question, however, is how large this cut-off should be chosen in
implementation.  To answer this question, we would need precise
estimates on the cut-off error.  Rigorous estimates for this cut-off
error are very difficult to obtain and the simplest estimates seem to
be, in any case, much larger than what we have seen in our
implementation. Instead, we have chosen to numerically estimate the
size of $\Omega$ by repeatedly increasing it until our FD solution
changed very little on the domain of interest $I \times J$.

% to review
Specifically, we have then taken $\Omega := [-3 , 3] \times [0.0193,
  1.6803]$ for the largest value of $\nu$ that we have tested ($\nu =
1.5$) and for the largest value of $T$ that we have tested ($T = 5$).
Some smaller intervals have been used for some of the smaller value of
$T$ and $\nu$ when this changed only marginally the results. Notice
that, on a logarithmic scale, the interval in $\sigma$ has been chosen
significantly larger than the corresponding interval in $x$ (when
compared to the intervals of interest), so the discretization in $\sigma$
required more nodes than the one in $x$.

We found that the method converged faster if the boundary condition $g$
was given by the solution of the Black-Scholes equation on the
boundary of $\Omega$, which corresponds to the solution of our PDE for
$\nu = 0$.

\subsubsection{Discretization errors}
% to review the numbers here: size of the sigma interval, nr of nodes ... 
For the discretization, be have begun with a uniform grid on the
interval $[-3, 3]$ with mesh size $h = 1/2$. However, on the interval
$[0.0193, 1.6803]$, we have chosen a {\em geometric progression} grid
with $19$ nodes (thus the end points of our interval of interest,
whose precise values are $0.1404$ and $0.2307$, are among the chosen
nodes). For $T = 5$, we have used $190$ time-steps for the first
iteration. We have used a direct method, because it is easy to
implement and reasonably efficient in our case, at least after
choosing the right grid in the $\sigma$ variable (i.e. a geometric
progression grid). We have then {\em successively refined} our grid by
dividing each interval into two smaller intervals, while preserving
the nature of the grids in each dimension. (That is, in the $x$
direction, we have chosen the arithmetic mean of the end points to
divide an old interval, whereas, in the $\sigma$ direction, we have
used the geometric mean.) Correspondingly, we had to multiply by 4 the
number of time steps each time when we performed a refinement in order
to satisfy the stability condition for the FD implementation. This has
thus lead us to a {\em sequence} $w_k$ of approximations of the
solution $w$ of Equation \eqref{eq.boxSABR}. In particular, each
approximate solution $w_k$ had required (essentially) 16 times more
time to compute than the previous one. This has imposed some stringent
limits on the precision of our approximations due to the running time
of the code (which had been written in C++ and run on a simple
laptop). The expected rate of convergence $\|w_{k+1} - w_k\| \approx
4^{-1} \|w_{k} - w_{k-1}\|$ has been observed almost immediately and
that has lead to good error estimates of the form $\|w - w_k\| \approx
3^{-1} \|w_{k} - w_{k-1}\|$ (unlike the cut-off error, which was
estimated empirically, although we knew that it decayed faster than
any exponential due to the decay of the Green function of $\pa_t
-L$). This rate of convergence was observed both on the domain of
interest $I \times J$ and on the total domain $\Omega$.

To see the kind of precision obtained, we included the
entry that begins with $1*$ in table \ref{table.FD1}. It corresponds
not to the last term $w_k$ in our sequence of FD approximations, but
rather to $w_{k-1}$. The entry corresponding to $w_k$ is right
underneath it, and we see that the numbers are very close (except the
estimated error $\|w -w_k\|_{I \times J}$, which is, as expected,
about four times smaller for the $k$th term). Notice that the values
recorded in that table represent 100 times the values of the corresponding
norms, since we were dealing with small numbers.

\subsubsection{Total error estimation}
Using the notation of the previous paragraphs, we see that the total
error (on the domain of interest) is
\begin{equation}\label{eq.cutoff}
  \|C_{SA} - w_k\|_{I \times J} \ \le \ \|C_{SA} - w\|_{I \times J} +
  \|w - w_k\|_{I \times J} \,.
\end{equation}
In this estimate, $\|C_{SA} - w\|_{I \times J}$ is the cut-off error,
due to our restricting the FD test to a {\em bounded domain}
$\Omega$. The cut-off error is independent of $k$ in our tests, so the
domain $\Omega$ had to be chosen rather large to start with. The term
$\|w - w_k\|_{I \times J}$ represents the error due to the FD
discretization and goes to 0 as expected as $c4^{-k}$.  The resulting
estimates for $\|C_{SA} - w_k\|_{I \times J}$ are included in the last
columns of Tables \ref{table.FD1} and \ref{table.FD2}. Let us explain
the meaning of the other columns in those tables.

Let $\| \, \cdot \, \| := \| \, \cdot \, \|_{\ell_2(I \times J)}$ and
$\| \, \cdot \, \|_{\infty} := \| \, \cdot \, \|_{\ell_\infty(I \times
  J)}$.  Let $C_H$ be the solution predicted by Hagan's model and
$C_{SA, 2}$ be the solution predicted by our Duhamel-Dyson series
model, as before. The ``log'' columns show the $\|\ln (C_{H}) -
\ln(w_k) \|$ and $\|\ln(C_{SA, 2}) - \ln(w_k) \|$ norms of the
differences at the nodal points, where $k$ corresponds to the last
discretization. We also write $\delta C_H := C_H - w_k$ and,
similarly, $\delta C_{SA, 2} := C_{SA, 2} - w_k$ and $\delta_{HD} :=
C_H - C_{SA, 2}$. We have performed our tests for $\rho = -.2$ and for
different values of the of the other parameters. The results are
included in Table \ref{table.FD1}. Similar results, but for $\rho =
-.5$ are included in Table \ref{table.FD2}, where we show separately
the estimates for the FD error and for the cut-off error. Because we
the numbers were very small, all norms in these tables were multiplied
by 100.
% trebuie sa folosesc noul program aici

\begin{table}
\caption{Comparison of the closed form approximate solutions with the
  iterative FD approximations and $\rho = -.2$. All norms were
  multiplied by 100.}
\begin{center}
\label{table.FD1}
\begin{tabular}{ | c | c | c | c | c | c | c | c | c | c | c | c |}
 \hline
 T \ & $\nu$ & $\|\delta C_{H} \|$ & $\| \delta C_{H} \|_{\infty}$ &
 log $C_H$ & $\|\delta C_{SA, 2} \|$ & $\| \delta C_{SA, 2}
 \|_{\infty}$ & log $C_{SA, 2}$ & $\|\delta_{HD}\|$ & est error \\
\hline 5 & \ 1 \ & \ 10.5 \ & \ 22.7 \ & \ 74.3 \ & \ 7.35 \ & \ 16.8
\ & \ 55.3 \ & \ 5.32 \ & \ .036 \ \\
%% 89.3875
%
\hline 2 & \ 1.5 \ & \ 6.32 \ & \ 14.2 \ & \ 78.  \ & \ 2.91 \ &
\ 7.72 \ & \ 56.7 \ & \ 5.9 \ & \ .0102 \ \\
%% 119.8205
%
\hline 2 & \ 1 \ & \ 1.45 \ & \ 3.28 \ & \ 35.3 \ & \ 0.939 \ & \ 2.26
\ & \ 38.6 \ & \ 1.72 \ & \ .0091 \ \\
\hline 2 & \ .5 \ & \ .1 \ & \ 0.23 \ & \ 4.03 \ & \ .136 \ & \ .398
\ & \ 6.68 \ & \ .179 \ & \ .0022 \ \\
%% 29.3532
%
\hline 1* & \ 1.5 \ & \ 1.25 \ & \ 2.81 \ & \ 43.9 \ & \ .741 \ &
\ 1.6 \ & \ 56.4 \ & \ 1.7 \ & \ .0442 \ \\
%% 93.3566
%
\hline 1 & \ 1.5 \ & \ 1.24 \ & \ 2.9 \ & \ 43.3 \ & \ .732 \ & \ 1.61
\ & \ 56.4 \ & \ 1.69 \ & \ .011 \ \\
%% 93.9172
%
\hline 1 & \ 1 \ & \ .241 \ & \ .53 \ & \ 14.3 \ & \ .2379 \ & \ .608
\ & \ 22.2 \ & \ 0.403 \ & \ .0101 \ \\
%% 48.9943
%
\hline .5 & \ 1\ & \ .029 \ & \ .088 \ & \ 2.13 \ & \ .051 \ & \ .179
\ & \ 4.83 \ & \ .0665 \ & \ .003 \ \\
%% 17.3963
%
\hline
\end{tabular}
\end{center}
\end{table}

The above tests (taking into account also the discretization and
cut-off errors) seem to suggest that the Duhamel-Dyson method is at
least as competitive as Hagan's method (if not better for large $T$
and parameters compatible with market data). The quantities that seem
the most relevant to support our conslusions are thus $\|\delta C_{H}
\|$ and $\|\delta C_{SA, 2} \|$.

The results of the tests in Tables \ref{table.FD1} and \ref{table.FD2}
sometimes allow us to compare the different other approximation
methods as follows. Let us do that for Hagan's formula and for our
second order approximation $C_{SA, 2}$. We have (all norms are on the
domain of interest $I \times J$)
\begin{multline}
 \| C_H - C_{SA, 2}\| \ge \| C_H - C_{SA} \| - \|C_{SA, 2} - C_{SA} \|
 \\
  \ge \, \| C_H - w_k \| - \|C_{SA, 2} - w_k \| - 2\|C_{SA} - w_k\|
 =: \|\delta C_{H} \| - \|\delta C_{SA, 2} \| - 2\|C_{SA} - w_k\| \,.
\end{multline}
This gives that $C_{SA, 2}$ will be closer to the actual solution
$C_{SA}$ (on the grid points of $I \times J$) whenever we have
$\|\delta C_{H} \| - \|\delta C_{SA, 2} \| > 2\|C_{SA} - w_k\|$. On
the other hand, if the norm $\| C_H - C_{SA, 2}\|$ is small, it will
be hard to decide which of $\| C_H - C_{SA} \|$ and $ \|C_{SA, 2} -
C_{SA} \|$ is smaller. This is the case when $T$ and $\nu$ are
small. In general, the difference $\| C_H - C_{SA, 2}\|$ goes faster
to 0 than $\|C_{SA} - w_k\|$, so for these values, it is more
difficult to tell which of $C_H$ or $C_{SA, 2}$ is closer to $
C_{SA}$. Thus, for small values of $\nu$ or $T$, we need more
precision (i.e. a larger $k$) in order to distinguish between models,
both of which seem to approximate very well the actual solution (for
parameters compatible with the market data, in particular, for $T \le
2$). Since increasing $k$ by 1 increases the time for running the
program approximately by a factor of 15, distinguishing $C_H$ and
$C_{SA,2}$ for small $T$ or $\nu$ is a true challenge in our tests.

\begin{table}
\medskip
\caption{Comparison of the closed form approximate solutions with the
  iterative FD approximations and $\rho = -.5$. All norms were
  multiplied by 100.}
\label{table.FD2}
\begin{center}
\begin{tabular}{ | c | c | c | c | c | c | c | c | c | c | c | c |}
 \hline
 T \ & $\nu$ & $\|\delta C_{H} \|$ & $\| \delta C_{H} \|_{\infty}$ &
 log $C_H$ & $\|\delta C_{SA, 2} \|$ & $\| \delta C_{SA, 2}
 \|_{\infty}$ & log $C_{SA, 2}$ & $\delta_{HD}$ & FD-err & cut-off \\
\hline
2 & \ 1 \ & \ 1.24 \ & \ 2.88 \ & \ 16.8 \ & \ 1.03 \ & \ 2.47 \ &
\ 39.4 \ & \ 1.54 \ & \ .0185 \ & \ .0418\ \\
\hline
1 & \ 1 \ & \ .253 \ & \ .543 \ & \ 2.4 \ & \ .314 \ & \ .752 \ &
\ 25.9 \ & \ .462 \ & \ .0212 \ & \ .0016\ \\
\hline
0.5 & \ 1 \ & \ .0497 \ & \ .137 \ & \ 1.41 \ & \ .0671 \ & \ .177 \ &
\ 13.7 \ & \ .0936 \ & \ .0058 \ & \ .0012 \ \\
\hline
\end{tabular}
\end{center}
\medskip
\end{table}

Another challenge in the FD implementation is that the cost of the
method increases very fast with $T$. If we could keep the cut-off
domain fixed (with $T$), the cost would grow linearly in $T$ (which is
already a challenge). However, as $T$ growth, we need to take a larger
$\Omega$, since the Green function decreases roughly like
$e^{-d^2/(2T)}$, where $d$ is the hyperbolic distance on $\RR \times
(0, \infty)$. To maintain the same precision for a larger $T$, we
would expect then the distance to $\pa \Omega$ to grow at least as
fast as $\sqrt{T}$. So if replace $T$ by $4T$, then we would have to
double the (hyperbolic) distance to $\pa M$. That would mean to
multiply $\sigma_{MAX}$ by roughly $e^{2 \nu}$. Since we had chosen a
geometric mesh in the $\sigma$ direction (exactly for this reason),
that would not increase too much the number of nodes in the $\sigma$
direction. (This number growth like $\ln( \Omega_{MAX})$.)  However,
it would affect the stability of the matrix in the (direct) FD method
by a factor of $e^{4 \nu}$ (because of the coefficient $\sigma^2$ in
front of $\pa_x^2$) and thus would require a proportional decrease in
the size of the time steps. The total number of time steps would hence
need to be thus increased by a factor of $4e^{4\nu}$. When $\nu$ is
large, this is very expensive.

See \cite{FouqueLorig, Les17, Siyan} for some recent results on the
mean-reverting case. We stress, however, than many theoretical results
that are true in the non-mean-reverting case (i.e. $\kappa_0 = 0$) are
not known and even may {\em not} true in the case $\kappa_0 \neq 0$.

%%%%%%%%%%%%%%%%%%%%%%%%%%%%%%%%%%%%%%%%%%%%%%%%%%%%%%%%%%%%
\section{Extensions of the method}
So far, we have completely solved a special case of the much more
general $\lambda$-SABR model. That is, we have taken $\beta=1$ and
$\kappa=0$. However, the applicability of the commutator approach does
not require $F(t)$ and $\sigma(t)$ to be log-normal. As mentioned
before, the key condition for this method to work is that the model
coefficients are all polynomial functions of the state variables. At
first glance, the SABR model with a general $\beta$ violates this
condition because the coefficient, $F(t)^{\beta}$, is not a polynomial
function of $F(t)$. Fortunately, we will show that, if the volatility
process is log-normal (i.e. $\kappa = 0$), a simple change of variable
transforms the SABR equations into one with polynomial coefficients.

The other assumption, that $\kappa = 0$, is not crucial either. When,
$\kappa\neq 0$, the mean-reverting drift term can be eliminated by
considering instead the forward volatility $\mathbb{E}_t[\sigma(T)]$
as our state variable. This creates slight complication in our model,
as the transformed SDE's will not be time-homogeneous. We will show
that this is not an essential issue for the application of the
commutator method.

With the above being said, the most general case, that $\kappa \neq 0$
and $\beta \neq 1$, is substantially more difficult and cannot be
dealt with by the commutator method alone. Therefore, in the rest of
the section, we consider only generalizations along the two
directions, 1) $\kappa \neq 0$ and $\beta =1$, 2) $\beta\neq 1$ and
$\kappa =0$.

\subsection{The log-normal model with mean-reverting volatility}
We consider first the simpler case that $\beta=1$, while the
volatility process, instead of being log-normal, is driven by a
process similar to an
Ornstein-Uhlenbeck process. Then our price-volatility dynamics under
the risk-neutral measure is assumed to be
\begin{equation}\label{meanrev-ln-SABR}
\begin{cases}
  \ dF(t) \, = \, \sigma(t)F(t) dW_1(t) & \\
  \ d\sigma(t) \, = \,\kappa (\theta-\sigma(t))dt \, + \nu
  \sigma(t)dW_2(t) \, & \\
 \ dW_1(t)dW_2(t)  \, = \,  \rho dt\, , &\\
\end{cases}
\end{equation} 
To eliminate the volatility drift term, we consider the
\textit{forward volatility}, defined as the time-$t$ conditional
expectation of the terminal volatility $\sigma(T)$,
$$z(t) = \EE_t[\sigma(T)].$$ 

Then $z(t)$ is a martingale by construction, and hence there is no
drift in its dynamics. To derive the exact relationship between $z(t)$
and $\sigma(t)$, we calculate
$$d(e^{\kappa t} \sigma(t)) \, =\, e^{\kappa t}(d\sigma(t) + \kappa
\sigma(t) dt) \, = \, \kappa \theta e^{\kappa t} dt + \nu
\sigma(t)e^{\kappa t} dW_2(t). $$ Ingrate the above from $t$ to $T$,
and take the conditional expectation $\EE_t[\cdot]$ on both sides, we
obtain
$$e^{\kappa T}z(t) - e^{\kappa t}\sigma(t) = \theta (e^{\kappa
  T}-e^{\kappa t}).$$ The above allows us to write the variable
$\sigma(t)$, as a function of our new state variable $z(t)$ and time,
as
\begin{equation}\label{sig-z}
\sigma(z,t) = e^{\kappa (T-t)}z(t) - \theta(e^{\kappa (T-t)}-1).
\end{equation}
It is easy to verify that
\begin{equation}\label{z-dynam}
dz(t)\, = \, \nu (z(t) - \theta(1-e^{\kappa (t-T)}))dW_2(t). 
\end{equation}
Therefore, we arrive at the transformed model,
\begin{equation}\label{trans-ln-SABR}
\begin{cases}
  \ dF(t) \, = \, \sigma(z,t)F(t) dW_1(t) & \\
  \ dz(t) \, = \, \nu e^{-\kappa(T-t)}\sigma(z,t)dW_2(t) \, & \\
 \ dW_1(t)dW_2(t)  \, = \,  \rho dt\, , &\\
\end{cases}
\end{equation}
We are now back at the previous case where the ``volatility'' process
has no drift, although the coefficients are now time-dependent
functions of the state variables. After replacing the price process
$F(t)$ with $X(t) = \ln{(F(t))}$, the option price $u(x,z,t)$ under
the transformed model solves the following initial value problem
\begin{equation}\label{pde-meanrev}
\begin{cases}
\begin{aligned}
& u_{\tau} = \frac{1}{2} \sigma^2 (u_{xx}-u_{x}) + \nu\rho
  e^{-\kappa\tau} \sigma^2 u_{xz} + \frac{1}{2}
  \nu^2e^{-2\kappa\tau}\sigma^2 u_{zz}\\ &u(x,z,0)=(e^x-K)^{+}\\
\end{aligned}
\end{cases}
\end{equation}

The corresponding differential operators are
\begin{equation}
L_0=\frac{1}{2}\sigma^2(\pa_x^2-\pa_x), \quad L_1=\rho
e^{-\kappa\tau}\sigma^2\pa_x\pa_z, \quad
L_2=\frac{1}{2}e^{-2\kappa\tau}\sigma^2\pa_z^2
\end{equation}
and 
$$L=L_0+\nu L_1+\nu^2 L_2 \quad \quad V=\nu L_1+\nu^2 L_2$$ Keep in
mind that $\sigma$ here is a function of $z$ and $\tau$ given by
\eqref{sig-z}, and the $L_i$'s are no longer
time-homogeneous. Therefore the fundamental solution of the system is
no longer a semi-group, but an evolution system, in the form of
$e^{\int_0^\tau L(s)ds}$. Repeated application of Duhamel's principle
yields the following Duhamel-Dyson perturbative series expansion of
$e^{\int_0^\tau L(s)ds}$

\begin{equation}
\begin{aligned}
e^{\int_0^\tau L(s)ds}&=e^{\int_0^\tau L_0(s)ds}+\int_0^\tau
e^{\int_{t_1}^\tau
  L_0(s)ds}V(t_1)e^{\int_0^{t_1}L_0(s)ds}dt_1\\ &+\int_0^\tau
\int_0^{t_1}e^{\int_{t_1}^\tau
  L_0(s)ds}V(t_1)e^{\int_{t_2}^{t_1}L_0(s)ds}V(t_2)e^{\int_0^{t_2}L_0(s)ds}dt_2dt_1\\
\end{aligned}
\end{equation}

Applying again the CHB identity, we can easily arrive at the following
lemma
\begin{lemma}
\label{meanrev-I-ops}
\begin{equation*}
\begin{aligned}
e^{\int_0^\tau L(s)ds}u(x,z,0)&=e^{\int_0^\tau L_0(s)ds}(1+\nu
J_1+\nu^2 J_2)u(x,z,0)\\ &=(1+\nu J_1+\nu^2 J_2)e^{\int_0^\tau
  L_0(s)ds}u(x,z,0)
\end{aligned}
\end{equation*}
with
\begin{equation*}
\begin{aligned}
  &J_1=\int_0^\tau[L_1(t_1),\int_0^{t_1}L_0(s)ds]dt_1\\
  &J_2=\int_0^\tau[L_2(t_1),\int_0^{t_1}L_0(s)ds] +
  \frac{1}{2}[\int_0^{t_1}L_0(s)ds,
    [\int_0^{t_1}L_0(s)ds,L_2(t_1)]]dt_1+I(L_1,L_1)\\ &I(L_1,L_1) =
  \int_0^\tau\int_0^{t_1}(L_1(t_1)+[L_1(t_1),
    \int_0^{t_1}L_0(s)ds])[L_1(t_2),\int_0^{t_2}L_0(s)ds]dt_2dt_1\\
\end{aligned}
\end{equation*}
\end{lemma}
\begin{proof}
The proof is a straightforward calculation using CHB identity and the
fact that $u(x,z,0)$ depends on $x$ only.
\end{proof}

\begin{lemma}
\label{meanrev-commutator}
The commutators are given by:
\begin{equation*}
\begin{aligned}
&[L_1(t_1),\int_0^{t_1}L_0(s)ds]=\frac{\rho}{2\kappa}(1-e^{-\kappa
    t_1})(ze^{\kappa t_1}-\theta(e^{\kappa t_1}-1))^2(e^{\kappa
    t_1}(z-\theta)+z+\theta)(\pa_x^3-\pa_x^2)\\ &[L_2(t_1),\int_0^{t_1}L_0(s)ds]
  = \frac{1}{2\kappa}(e^{-\kappa t_1}-e^{-2\kappa t_1})(ze^{\kappa
    t_1}-\theta(e^{\kappa t_1}-1))^2(e^{\kappa
    t_1}(z-\theta)+z+\theta)(\pa_x^2-\pa_x)\pa_z\\ &\hspace{+3.2cm} +
  \frac{1}{4\kappa}(1-e^{-2\kappa t_1})(ze^{\kappa
    t_1}-\theta(e^{\kappa t_1}-1))^2(\pa_x^2-\pa_x)
  \\ &[\int_0^{t_1}L_0(s)ds,[\int_0^{t_1}L_0(s)ds,L_2(t_1)]] =
  \frac{1}{4\kappa^2}e^{-2\kappa t_1}(e^{\kappa t_1}-1)^2(ze^{\kappa
    t_1}-\theta(e^{\kappa t_1}-1))^2\\&\ \hspace{+5.25cm}\times
  (e^{\kappa t_1}(z-\theta)+z+\theta)^2(\pa_x^2-\pa_x)^2
\end{aligned}
\end{equation*}
\end{lemma}
\begin{proof}
Let us first calculate the time-indexed commutator $[L_1(t_1), \,
  L_0(t_0)]$. For brevity, we use the notation $v_i$ to denote
$\sigma(z,t_i)^2$, and use $v_i'$ to denote $\frac{\pa\,
  \sigma(z,t_i)^2}{\pa\, z}$. Recall that $\sigma(z,\tau) = e^{\kappa
  \tau}z-\theta(e^{\kappa\tau}-1)$. As a result,
$$v_i' = 2 (e^{\kappa t_i}z-\theta(e^{\kappa t_i}-1))e^{\kappa t_i}$$
$$[L_1, L_0] = \frac{1}{2}\rho e^{-\kappa
  t_1}v_1v_0'(\pa_x^3-\pa_x^2).$$
Hence $$[L_1(t_1),\int_0^{t_1}L_0(s)ds] = \rho e^{-\kappa t_1}v_1
\int_0^{t_1} (e^{\kappa t_0}z-\theta(e^{\kappa t_0}-1))e^{\kappa
  t_0}dt_0(\pa_x^3-\pa_x^2).$$ Rearranging the terms yields the first
identity. Calculations for the second and third identities are
similar.
\end{proof}
Combing the above two lemmas, $J_1$ and $J_2$ can be obtained through
straightforward (but lengthy) calculations. Then, same as before, the
operator $1+\nu J_1+\nu^2 J_2$ can be written as a linear combination
of $\pa_x$, $\pa_x^2$, $\ldots$, $\pa^6_x$, with each coefficient
being elementary functions of parameters. The only question left is,
how do we calculate the transition density
$e^{\int_t^TL_0(s)ds}(x,t;y,T)$, in the presence of volatility
mean-reversion?

Note that when $\nu=0$, we arrive at the
Black-Scholes model with deterministic volatility, given by
\begin{equation}\label{sig-ODE}
\frac{d\sigma(t)}{dt} = \kappa(\theta - \sigma(t)).
\end{equation}
Hence the terminal log-price, $X(T)$ can be solved explicitly as,
$$X(T) = X(t)-\frac{1}{2}\int_t^T\sigma(s)^2ds +
\int_t^T\sigma(s)dW_1(s).$$ Conditional on $\mathcal{F}_t$, $X(T)$ is
normally distributed with mean $X(t) - \frac{1}{2}\overline{V}(t)$ and
variance $\overline{V}(t)$, with $\overline{V}(t)$ denote the time
aggregated variance. Then the transition density in the case of
$\nu=0$ is given by
\begin{equation}\label{meanrev-density}
e^{\int_t^TL_0(s)ds}(x,t;y,T)= \frac{1}{ \sqrt{2\pi \overline{V}(t)}}
e^{-\frac{(y-x+\frac{1}{2}\overline{V}(t))^2}{2 \overline{V}(t)}}
\end{equation} 
Here $\overline{V}(t)$ can be obtained by solving the ODE
\eqref{sig-ODE}, as
\begin{equation}\label{total-var}
\overline{V}(t) = \theta^2\tau + \frac{2}{\kappa} \theta(z(t) -
\theta)(e^{\kappa\tau}-1) +
\frac{1}{2\kappa}(z(t)-\theta)^2(e^{2\kappa\tau}-1)
\end{equation}
$e^{\int_t^TL(s)ds}(x,t;y,T)$ as well as the option price follow from
lemmas \ref{meanrev-I-ops} and \ref{meanrev-commutator} and equation
\ref{meanrev-density}.

\subsection{The SABR model (with general $\beta$)}
We now consider the original SABR model with the volatility following
a pure log-normal process, i.e.
\begin{equation}\label{orig-SABR}
\begin{cases}
  \ dF(t) \, = \, \sigma(t)F(t)^\beta dW_1(t) & \\
  \ d\sigma(t)  \, = \, \nu \sigma(t)dW_2(t)  \,  & \\
 \ dW_1(t)dW_2(t)  \, = \,  \rho dt\, , &\\
\end{cases}
\end{equation}
With $\beta\neq 1$, the model can no longer be thought of as the
Black-Scholes model perturbed by the vol-of-vol coefficient. Indeed,
when $\nu=0$, \eqref{orig-SABR} degenerates to the CEV model which is
itself hard to solve. Therefore, our first objective is to apply a
change of variable technique such that the transformed model becomes
again a perturbed Black-Scholes model.

First notice that the
``instantaneous volatility'' is not $\sigma(t)$, but
$\sigma(t)F(t)^{\beta-1}$. We call it the \textit{local volatility
  process} and denote it by $M(t)$. The good news that, when
$\sigma(t)$ follows a driftless log-normal process as in
\eqref{orig-SABR}, the Ito-differential of $M(t)$ does not depend on
$\sigma(t)$ and $F(t)$ individually, but on $M(t)$ alone. To see this,
apply Ito's lemma
\begin{equation}
\begin{split}
dM(t)\, = \, (\beta-1)M(t)^2(\nu\rho+\frac{1}{2}(\beta-2)M(t))dt+
(\beta-1)M(t)^2dW_1(t)+\nu M(t)dW_2(t)
\end{split}
\end{equation} 
Let $\epsilon$ denote $\beta-1$, the SABR model is now rewritten as
\begin{equation}\label{trans-SABR}
\begin{cases}
  \ dF(t) \, = \, M(t)F(t) dW_1(t) & \\
  \ dM(t) \, = \, \epsilon
  M(t)^2(\nu\rho+\frac{1}{2}(\epsilon-1)M(t))dt+ \epsilon
  M(t)^2dW_1(t)+\nu M(t)dW_2(t) \, & \\
 \ dW_1(t)dW_2(t)  \, = \,  \rho dt\, , &\\
\end{cases}
\end{equation}
Therefore, we can now consider the SABR model as the Black-Scholes
model perturbed by two ``vol-of-vol'' parameters, $\epsilon$ and
$\nu$. Moreover, the coefficients in the above equation are all
polynomial functions of $F(t)$ and $M(t)$. As a result, if we perform
a Taylor series expansion for the option price, $C_{SA}(f,m)$ on both
$\epsilon$ and $\nu$, the CHB theorem guarantees that the coefficients
can be calculated exactly with finitely many terms.

Note that the
above result no longer holds when the volatility is mean-reverting, as
the local volatility process $M(t)$ is no longer autonomous. We also
point out that, unlike in the previous case with mean-reverting
volatility, we do not try to get rid of the ``volatility drift'' term
in the $M(t)$-equation. The reason being that we are now performing
Taylor expansion on both $\epsilon$ and $\nu$, hence the drift of
$M(t)$ does not enter our zeroth order operator and complicate the
rest of the calculations.

Equation \eqref{trans-SABR} implies that
the option price, $u(x,m,\tau)$ solves the following initial value
problem
\begin{equation}
\begin{cases}
&u_\tau = \frac{1}{2}m^2(u_{xx}-u_x)+\epsilon
  m^2(\nu\rho+\frac{1}{2}(\epsilon-1)m)u_m+m^2(\epsilon m+
  \nu\rho)u_{xm} \\ &\hspace{+0.8cm}
  +\frac{1}{2}m^2(\epsilon^2m^2+2\epsilon \nu\rho m +
  \nu^2)u_{mm},\\ &u(x,m,0) = (e^{x} - K)^+.
\end{cases}
\end{equation}
The second order differential operator on the right hand side can be
written as,
$$L = L_0 + \epsilon L_1^{\epsilon} + \nu L_{1}^\nu + \epsilon^2
L_2^\epsilon + \epsilon\nu L_2^{\epsilon\nu} + \nu^2 L_2^\nu,$$ where
\begin{equation*}
\begin{split}
&L_0 = \frac{1}{2}m^2(\pa_x^2-\pa_x), \\ &L_1^\epsilon =
  \frac{1}{2}(\epsilon-1)m^3\pa_m + m^3\pa_x\pa_m,\quad L_1^\nu = \rho
  m^2 \pa_x\pa_m,\\ &L_2^\epsilon = \frac{1}{2}m^4\pa_m^2, \quad
  L_2^\nu = \frac{1}{2}m^2\pa_m^2, \quad L_2^{\epsilon\nu} = \rho
  m^2(\pa_m + m \pa^2_m).
\end{split}
\end{equation*}

%%%%%%%%%%%%%%%%%%%%%%%%%%%%%%%%%%%%%%%%%%%%%%%%%%%%%%%%%%%%

%====================BIBLIOGRAPHY============================

\end{document}